\documentclass[acmsmall,nonacm,screen]{acmart}

\setcopyright{none}                 
\renewcommand\footnotetextcopyrightpermission[1]{}  
\usepackage{amsmath,amsthm}

\newtheorem{theorem}{Theorem}
\newtheorem{lemma}[theorem]{Lemma}

\theoremstyle{definition}
\newtheorem{definition}[theorem]{Definition}
\newtheorem{example}[theorem]{Example}

\usepackage[dvipsnames,svgnames]{xcolor}
\hypersetup{colorlinks=true, citecolor=DodgerBlue}
\usepackage[inline]{enumitem}
\usepackage{multirow}
\usepackage{multicol}
\usepackage{tabu}
\usepackage{longtable}
\newcolumntype{L}{>{$}l<{$}} 
\newcolumntype{C}{>{$}c<{$}} 
\newcolumntype{R}{>{$}r<{$}}
\newcounter{rownum}
\newcommand{\rownum}{\stepcounter{rownum}\therownum}

\usepackage{subcaption}
\DeclareCaptionFormat{custom}{
    \mbox{\sffamily\upshape\bfseries {#1 #2}} {#3}
}
\usepackage{ebproof}
\ebproofset{
    rule thickness=0.7pt,
    left label template=\footnotesize$\inserttext$,
    right label template=\footnotesize$\inserttext$,
    }

\usepackage{algorithm}
\usepackage[indLines=false]{algpseudocodex}
\usepackage{etoolbox}   
\algnewcommand\algorithmicswitch{\textbf{switch}}
\algnewcommand\algorithmiccase{\textbf{case}}
\algnewcommand\algorithmicdefaultcase{\textbf{default}}
\makeatletter
\algdef{SE}[SWITCH]{Switch}{EndSwitch}[1]{%
	\algpx@startCodeCommand\algpx@startIndent\algorithmicswitch\ #1%
}{%
	\algpx@startEndBlockCommand\algpx@endIndent\algorithmicend\ \algorithmicswitch%
}
\algdef{SE}[CASE]{Case}{EndCase}[1]{%
	\algpx@startCodeCommand\algpx@startIndent\algorithmiccase\ #1\textbf{:}%
}{%
	\algpx@startEndBlockCommand\algpx@endIndent\algorithmicend\ \algorithmiccase%
}
\algdef{SE}[DEFAULTCASE]{Default}{EndDefault}{%
	\algpx@startCodeCommand\algpx@startIndent\algorithmicdefaultcase\textbf{:}%
}{%
	\algpx@startEndBlockCommand\algpx@endIndent\algorithmicend\ \algorithmicdefaultcase%
}
\pretocmd{\Switch}{\algpx@endCodeCommand}{}{}
\pretocmd{\EndSwitch}{\algpx@endCodeCommand[0]}{}{}
\pretocmd{\Case}{\algpx@endCodeCommand}{}{}
\pretocmd{\EndCase}{\algpx@endCodeCommand[0]}{}{}
\pretocmd{\Default}{\algpx@endCodeCommand}{}{}
\pretocmd{\EndDefault}{\algpx@endCodeCommand[0]}{}{}
\makeatother
\algtext*{EndCase}

\renewcommand{\phi}{\varphi}
\renewcommand{\epsilon}{\varepsilon}
\renewcommand{\le}{\leqslant}
\renewcommand{\ge}{\geqslant}

\newcommand{\defemph}[1]{{\rmfamily\itshape \color{MediumVioletRed} {#1}}}
\newcommand{\mathsc}[1]{\mbox{\rmfamily\mdseries\upshape\scshape {#1}}}
\newcommand{\closedAss}[2]{\infer0[{#1}]{\hspace{0.35em}{#2}\hspace{0.35em}}}

\newcommand{\taggedAbs}[2]{\infer[left label={#1}]1[\rnabs{}]{{#2}}}
\newcommand{\taggedCase}[2]{\infer[left label={#1}]3[\rncase]{{#2}}}

\newcommand{\taggedRaa}[2]{\infer[left label={#1}]1[\rnraa{}]{{#2}}}

\newcommand{\impl}{\rightarrow}
\newcommand{\disj}{\vee}
\newcommand{\conj}{\wedge}
\newcommand{\atoms}{\mathsf{PV}}

\newcommand{\formsize}[1]{\vert{#1\vert}}

\newcommand{\intlogic}{\mathsf{NJ}}
\newcommand{\classlogic}{\mathsf{NK}}
\newcommand{\rnax}[1]{\mathsc{ass}_{{#1}}}
\newcommand{\rnemp}[1]{\mathsc{emp}_{{#1}}}
\newcommand{\rnraa}[1]{\mathsc{raa}_{{#1}}}
\newcommand{\rnabs}[1]{\mathsc{abs}_{{#1}}}
\newcommand{\rnapp}{\mathsc{app}}
\newcommand{\rnpair}{\mathsc{pair}}
\newcommand{\rnfst}{\mathsc{fst}}
\newcommand{\rnsnd}{\mathsc{snd}}
\newcommand{\rnleft}[1]{\mathsc{left}_{{#1}}}
\newcommand{\rnright}[1]{\mathsc{right}_{{#1}}}
\newcommand{\rncase}{\mathsc{case}}
\newcommand{\rnrule}{\mathsc{r}}
\newcommand{\sfof}{\mathsf{sf}}
\newcommand{\clof}{\mathsf{cl}}
\newcommand{\lrof}{\mathsf{lr}}
\newcommand{\assof}{\mathsf{ass}}
\newcommand{\concof}{\mathsf{conc}}
\newcommand{\sizeof}{\mathsf{sz}}
\newcommand{\degof}{\mathsf{deg}}
\newcommand{\spof}{\mathsf{sp}}
\newcommand{\formsof}{\mathsf{forms}}
\newcommand{\intDerives}{\vdash_{\mathsf{int}}}

\newcommand{\classDerives}{\vdash_{\mathsf{c}}}

\newcommand{\graft}[3]{{#1}[{#2} \mapsto {#3}]}
\newcommand{\cutrank}{\mathsf{cr}}
\newcommand{\rankadd}{\ast}
\newcommand{\degree}{\mathsf{deg}}
\newcommand{\cutform}{\mathsf{cf}}
\newcommand{\cutwt}{\mathsf{wt}}
\newcommand{\maxdeg}{\mathsf{maxdeg}}
\newcommand{\contractum}{\mathsf{ctm}}
\newcommand{\oneStep}{\rightsquigarrow}
\newcommand{\os}{\textsc{step}}
\newcommand{\eliminator}{\mathsf{E}}
\newcommand{\elimbot}{\mathsf{R}}
\newcommand{\elimneg}[1]{\mathsf{N}({{#1}})}
\newcommand{\elimapp}[2]{\mathsf{A}({{#1}}, {{#2}})}
\newcommand{\elimfst}[1]{\mathsf{F}({{#1}})}
\newcommand{\elimsnd}[1]{\mathsf{S}({{#1}})}
\newcommand{\elimcase}[2]{\mathsf{C}({{#1}}, {{#2}})}
\newcommand{\elimtrans}[2]{[{#1}]{#2}}

\title{{New Proofs of Weak Normalization for Propositional Logic}} 
\author{S P Suresh}
\email{spsuresh@cmi.ac.in}
\affiliation{%
  \institution{Chennai Mathematical Institute}
  \city{Chennai}
  \country{India}
}
\date{}

\begin{document}
\begin{abstract}
    We present new proofs of weak normalization for intuitionistic and classical propositional logics (with the full set of operators $\{\bot, \impl, \conj, \disj\}$). These proofs work with cuts rather than cut segments, and they provide explicit ``local'' rules for determining whether to contract a whole proof or reduce one of its subproofs, and in the latter case, which subproof to reduce. Interestingly, much of the complication in the case of intuitionistic logic is due to the disjunction elimination rule, while our version of the same rule for classical logic has $\bot$ as conclusion always, and so is much easier to handle. All the complication in the case of classical logic shifts to cuts involving $\rnraa{}$, the \emph{reductio ad absurdum} rule. We also discuss a formalization of the entire proof in Lean, and present a deterministic algorithm for weak normalization.
\end{abstract}
\maketitle
\raggedbottom

\section{Introduction}\label{sec:intro}
\defemph{Natural Deduction} is a fundamental framework for logical proofs, with an extremely rich history. It was introduced by Gerhard Gentzen (and independently by Stanis{\l}aw Ja{'s}kowski) in the 1930s~\cite{Gentzen1935a,Gentzen1935b,Gentzen1969,Jaskowski1934}. Gentzen also introduced the notions of \defemph{Normalization} and \defemph{Subformula Property}, and introduced the allied system of \defemph{Sequent Calculus} and proved the fundamental \defemph{Cut Elimination Theorem}, which is closely related to normalization. He used this technical machinery to provide an astounding proof of the consistency of Peano arithmetic. These systems have been studied extensively for nearly a century. Further advances in Natural Deduction were by Dag Prawitz, who provided direct proofs of weak and strong normalization for Natural Deduction~\cite{Prawitz1965,Prawitz1971}. These results are so fundamental that they form an integral part of many introductory textbooks on proof theory -- see \cite{GirardLafontTaylor1989, TroelstraSchwichtenberg2000, NegriVonPlato2001, MancosuGalvanZach2021}, for example.

Another important strand of research concerns the \defemph{Curry-Howard correspondence}~\cite{Howard1980}, which establishes a deep connection between intuitionistic logic and typed $\lambda$-calculus. This is also the focus of many textbooks~\cite{GirardLafontTaylor1989, SorensenUrzyczyn2006}. This strand has culminated in the development and widespread use of \defemph{proof assistants} like Rocq~\cite{CoquandHuet1988, RocqWebsite} and Lean~\cite{LeanWebsite}.

\subsection{Motivation for this work}
The key idea of weak normalization is quite intuitive for the fragment of the logic with only implication (or equivalently, simply typed $\lambda$-calculus with only abstraction and application). The idea is to identify subproofs (equivalently, subexpressions) called \defemph{cuts} or \defemph{redexes}. These are subproofs where an implication is eliminated immediately after being introduced (equivalently, expressions of the form $(\lambda{x}\cdot{M})\,N$). We then associate a \defemph{degree} to each cut, and repeatedly simplify the rightmost innermost cut of maximal degree in a proof. This simplification increases the size of the proof in most cases, but we identify a different measure (the \defemph{cut-rank} of the proof) which decreases at each simplification step. 

This template also works when we add conjunction to the syntax, but disjunction complicates the proof considerably. One is forced to work with \defemph{cut-segments} instead of cuts. Rather than simplifying the rightmost innermost cut of maximum degree at each step, one has to identify an appropriate cut segment to reduce next. It is quite hard to precisely state the conditions that identify the next cut segment to be reduced, and the proof becomes quite unwieldy. 

To our knowledge, texts treat the proof of weak normalization for full natural deduction in three ways:
\begin{enumerate}
    \item Present the full proof for the fragment with only implication and conjunction, and present the additional permutative conversions needed for the full proof, without presenting an explicit reduction strategy.
    \item Prove \emph{strong normalization} of a rich framework like System F, and derive weak normalization for various systems (including full propositional intuitionistic logic) as a consequence. 
    \item Present the original elaborate proofs involving cut-segments. This style of treatment can be seen in the recent excellent text~\cite{MancosuGalvanZach2021} by Mancosu, Galvan and Zach.
\end{enumerate}

\subsection{Our contribution}
In this paper, we present reasonably short and simple proofs of weak normalization for both $\intlogic$ (Natural Deduction for intuitionistic propositional logic) and $\classlogic$ (Natural Deduction for classical propositional logic). We have made an attempt to obtain inductive definitions of all the key concepts and transformations. The key technical contributions are the following:
\begin{itemize}
    \item We present the proof only in terms of cuts, rather than cut-segments.
    \item We present a new definition of cut rank for the case when the main subproof of a cut ends in $\disj$-elimination. 
    \item We define a notion of one-step reduction which encodes the normalization strategy. Importantly, this definition is ``local'' -- we determine whether to contract the whole proof at the root or reduce one of the immediate subproofs, and in the latter case, which subproof to reduce, by checking how various (numerical) parameters associated with the proof and immediate subproofs are related to each other. The elaborate presentation in~\cite{MancosuGalvanZach2021} was of immense help to us in deriving the conditions that determine the next reduction. 
    \item For classical logic, we consider a version of $\disj$-elimination which always has $\bot$ as conclusion. We can derive the standard $\disj$-elimination from the restricted version and $\rnraa{}$.
    \item All the complexity in the normalization proof for classical logic shifts to cuts involving $\rnraa{}$, and we handle them in the same inductive fashion by generalizing the reduction rules for $\lambda\mu$-calculus (as presented in Chapter~6 of~\cite{SorensenUrzyczyn2006}). 
\end{itemize}

Since this is a new approach to the proof of a fundamental theorem, we ensured that the approach is correct by formalizing both the proofs in Lean.\footnote{The formalization involved 2--3 months of intermittent effort by the author, and consists of about 5000 lines of code.} A particularly important aspect of the formalization is this -- apart from proving the existence of a normal derivation corresponding to each derivation, we also transcribe our rules for reduction into a terminating, deterministic \emph{normalization algorithm}. The formalization is hosted at~\cite{NatDed}. 

\subsection{Structure of the paper} 
We present the basic definitions for the system $\intlogic$ in Section~\ref{sec:nj}, and discuss the challenges associated with normalization in Section~\ref{sec:challenge}. We then present our simplified proof for weak normalization in Section~\ref{sec:njNorm}. We then briefly discuss aspects of our Lean formalization in Section~\ref{sec:formalization}. The differences in the case of classical logic, and the new challenges that arise are discussed in Section~\ref{sec:nk}. The normalization proof for classical logic is presented in Section~\ref{sec:nkNorm}. We end with some concluding remarks. 

\section{\texorpdfstring{$\intlogic$}{\textsf{NJ}}: Natural Deduction for Intuitionistic Propositional Logic}\label{sec:nj}
Assume a countably infinite set $\atoms$ of \defemph{propositional variables}. The set of all formulas $\Phi$ is specified in BNF as follows:
\[
    \alpha, \beta \in \Phi \coloneq p\ \mid\ \bot\ \mid\ (\alpha\impl\beta)\ \mid\ (\alpha\conj\beta)\ \mid\ (\alpha\disj\beta).
\]
We use $p, q, r$ and variants to denote propositional variables, $\alpha, \beta, \gamma, \delta$ and variants to denote arbitrary formulas, and use $X, Y, Z$ and variants to denote arbitrary sets of formulas. As usual, $\neg\alpha$ abbreviates $\alpha\impl\bot$. 

The set of \defemph{subformulas} of $\alpha$, denoted $\sfof(\alpha)$, is the least set $X$ containing $\alpha$ s.t.\ $\{\beta,\gamma\} \subseteq X$ whenever $\beta\circ\gamma \in X$. The \defemph{size} of a formula $\alpha$, denoted $\formsize{\alpha}$, is $1$ if $\alpha \in \atoms \cup \{\bot\}$, and $1+\formsize{\beta}+\formsize{\gamma}$ if $\alpha = \beta\circ\gamma$. For a set of formulas $X$, $\sfof(X) \coloneq \bigcup_{\alpha \in X}\,\sfof(\alpha)$.

\begin{definition}[Proofs]
    The family of proofs $\{\Pi_{\alpha} \mid \alpha \in \Phi\}$ is defined by simultaneous induction as follows, where for each $\alpha \in \Phi$ and $\pi \in \Pi_{\alpha}$, $\pi$ is a proof with conclusion $\alpha$.
    \begin{itemize}[nosep]
        \item For all $\alpha \in \Phi$, $\rnax{\alpha} \in \Pi_{\alpha}$. 
        \item If $\mu \in \Pi_{\bot}$, then $\rnemp{\alpha}(\mu) \in \Pi_{\alpha}$.
        \item If $\mu \in \Pi_{\beta}$, then $\rnabs{\alpha}(\mu) \in \Pi_{\alpha\impl\beta}$.
        \item If $\mu \in \Pi_{\alpha\impl\beta}$ and $\nu \in \Pi_{\alpha}$, then $\rnapp(\mu,\nu) \in \Pi_{\beta}$.
        \item If $\mu \in \Pi_{\alpha}$ and $\nu \in \Pi_{\alpha}$, then $\rnpair(\mu,\nu) \in \Pi_{\alpha\conj\beta}$. 
        \item If $\mu \in \Pi_{\alpha\conj\beta}$, then $\rnfst(\mu) \in \Pi_{\alpha}$ and $\rnsnd(\mu) \in \Pi_{\beta}$. 
        \item If $\mu \in \Pi_{\alpha}$ and $\nu \in \Pi_{\beta}$, then $\rnleft{\beta}(\mu), \rnright{\alpha}(\nu) \in \Pi_{\alpha\disj\beta}$. 
        \item If $\chi \in \Pi_{\alpha\disj\beta}$ and $\mu,\nu \in \Pi_{\delta}$, then $\rncase(\chi,\mu,\nu) \in \Pi_{\delta}$. 
    \end{itemize}
\end{definition}
The above rules for constructing proof can also be presented succinctly in the form given by the following table, where we list the conditions under which we can define a new proof $\pi$, assuming that $\mu, \nu, \chi$ are proofs already defined. We simultaneously define $\concof(\pi)$, the \defemph{conclusion} of $\pi$.
\begin{longtable}{R|C|C}
    \pi & \text{Preconditions} & \concof(\pi) \\
    \hline
    \rnax{\alpha} & \text{None} & \alpha \\
    \rnemp{\alpha}(\mu) & \concof(\mu) = \bot & \alpha \\ 
    \rnabs{\alpha}(\mu) & \concof(\mu) = \beta & \alpha\impl\beta \\ 
    \rnapp(\mu,\nu) & \concof(\mu) = \alpha\impl\beta, \concof(\nu) = \alpha & \beta  \\
    \rnpair(\mu,\nu) & \concof(\mu) = \alpha, \concof(\nu) = \beta & \alpha\conj\beta \\
    \rnfst(\mu) & \concof(\mu) = \alpha\conj\beta & \alpha \\
    \rnsnd(\mu) & \concof(\mu) = \alpha\conj\beta & \beta \\
    \rnleft{\beta}(\mu) & \concof(\mu) = \alpha & \alpha\disj\beta \\
    \rnright{\alpha}(\mu) & \concof(\mu) = \beta & \alpha\disj\beta \\
    \rncase(\chi,\mu,\nu) & \concof(\chi) = \alpha\disj\beta, \concof(\mu) = \concof(\nu) = \delta & \delta \\
\end{longtable}
We define $\assof(\pi)$, the \defemph{set of assumptions} of $\pi$, in the table below. We use $X-\alpha$ etc.\ as shorthand for $X\setminus\{\alpha\}$. Since we can determine precisely the set of assumptions of any proof, we do not need to consider sequents in proofs. Each node is labelled only by formulas. 

\begin{longtable}{R|C|C}
    \pi & \text{Preconditions} & \assof(\pi) \\
    \hline
    \hline
    \rnax{\alpha} & \text{None} & \{\alpha\} \\
    \hline 
    \rnabs{\alpha}(\mu) & \assof(\mu) = X & X-\alpha \\ 
    \hline
    \rnemp{\alpha}(\mu), \rnfst(\mu), \rnsnd(\mu), & \multirow{2}{*}{$\assof(\mu) = X$} & \multirow{2}{*}{$X$} \\
    \rnleft{\beta}(\mu), \rnright{\alpha}(\mu) && \\ 
    \hline
    \rnapp(\mu,\nu), \rnpair(\mu,\nu) & \assof(\mu) = X, \assof(\nu) = Y & X \cup Y \\
    \hline
    \multirow{2}{*}{$\rncase(\chi,\mu,\nu)$} & \assof(\chi) = X, \concof(\chi) = \alpha\disj\beta, & \multirow{2}{*}{$X \cup (Y-\alpha) \cup (Z-\beta)$} \\ 
    & \assof(\mu) = Y,\ \assof(\nu) = Z & \\
\end{longtable}

\begin{example}\label{ex:proofTerms}
    In these examples, we write $\alpha_{1}, \ldots, \alpha_{k} \vdash \beta$ to mean that $\{\alpha_{1}, \ldots, \alpha_{k}\} \vdash \beta$.
    \begin{enumerate}[nosep]
        \item $\rnabs{\alpha}(\rnax{\alpha})$ is a proof of $\vdash \alpha\impl\alpha$. 
        \item $\rnabs{\alpha}(\rnabs{\beta}(\rnax{\alpha}))$ is a proof of $\vdash \alpha \impl \beta \impl \alpha$. 
        \item $\rnabs{\alpha}(\rnabs{\beta}(\rnabs{\alpha}(\rnax{\beta})))$ is a proof of $\vdash \alpha \impl \beta \impl \alpha \impl \beta$.
        \item $\rnpair(\rnax{\alpha}, \rnax{\beta})$ is a proof of $\alpha, \beta \vdash \alpha\conj\beta$.
        \item The following is a proof of $\gamma \vdash \alpha\impl(\beta\disj\alpha)$:
        {
            \setlength{\abovedisplayskip}{1pt}
            \setlength{\belowdisplayskip}{1pt}
            \[ 
                \rnabs{\alpha}(\rncase(\rnleft{\beta}(\rnpair(\rnax{\alpha},\rnax{\gamma}))),\ \rnright{\beta}(\rnfst(\rnax{\alpha\conj\gamma})),\ \rnleft{\alpha}(\rnax{\beta})).
            \]
        }
    \end{enumerate}
\end{example}

We use $\pi, \mu, \nu, \chi, \rho, \upsilon$ and variants to denote arbitrary proofs. If $\pi$ is a proof with $\assof(\pi) = X$ and $\concof(\pi) = \alpha$, we say that $\pi$ is a proof of $X \vdash \alpha$. We say that $\alpha$ is \defemph{provable (or derivable) from assumptions} $X$, and denote it by $X \intDerives \alpha$ if there is a proof of $Y \vdash \alpha$ for some $Y \subseteq X$. We say that $\alpha$ is \defemph{provable} (in symbols: $\intDerives \alpha$) if $\emptyset \intDerives \alpha$. 

For a proof $\pi = \rnrule(\pi_{1}, \ldots, \pi_{k})$, we define the \defemph{last rule} of $\pi$, denoted $\lrof(\pi)$, to be $\rnrule$. For a proof $\pi$, we denote by $\sizeof(\pi)$, $\spof(\pi)$, and $\formsof(\pi)$, respectively, the \defemph{size of} $\pi$, \defemph{subproofs of} $\pi$, and \defemph{the set of formulas occurring in} $\pi$. They are defined in the table below. 
\begin{longtable}{R|C|C|C}
    \pi & \sizeof(\pi) & \spof(\pi) & \formsof(\pi) \\
    \hline
    \rnax{\alpha} & 1 & \{\pi\} & \{\alpha\} \\
    \rnrule(\pi_{1}, \ldots, \pi_{k}) & 1 + \sizeof(\pi_{1}) + \cdots + \sizeof(\pi_{k}) & \{\pi\} \cup \spof(\pi_{1}) \cup \cdots \cup \spof(\pi_{k}) & \{\concof(\mu) \mid \mu \in \spof(\pi)\} \\ 
\end{longtable}
A \defemph{proper subproof of $\pi$} is any $\mu \in \spof(\pi)\setminus\{\pi\}$. For $\pi = \rnrule(\pi_{1}, \cdots, \pi_{k})$, its \defemph{immediate subproofs} are $\{\pi_{1}, \ldots, \pi_{k}\}$. If $\rnrule \in \{\rnemp{}, \rnapp, \rnfst, \rnsnd, \rncase\}$, then $\pi_{1}$ is its \defemph{main subproof}. An immediate subproof of $\pi$ that is not main is an \defemph{ancilliary subproof}.

\subsection{Proof Trees}~\label{sec:proofTrees}
We have presented proofs as \defemph{proof terms}, which helps make the presentation quite precise and easy to transcribe in a functional language like Lean. While ideal for a formal study of proofs and for presenting various constructions of proofs, proof terms are not very easy to parse, as can be seen from the ones presented in Example~\ref{ex:proofTerms}. The inherent tree structure of proofs is obscured by this presentation. Hence we use \defemph{proof trees} in examples, and while discussing the structure of proofs and constructions on them.  

\begin{example}~\label{ex:proofTrees}
    Presented below are proof trees corresponding to the proof terms in Example~\ref{ex:proofTerms}. 
    \begin{enumerate}
        \item A proof of $\vdash \alpha\impl\alpha$.
        \[
            \begin{prooftree}
                \closedAss{x}{\alpha}
                \taggedAbs{x}{\alpha\impl\alpha}
            \end{prooftree}
        \]
        \item A proof of $\vdash \alpha\impl\beta\impl\alpha$.
        \[
            \begin{prooftree}
                \closedAss{x}{\alpha}
                \taggedAbs{}{\beta\impl\alpha}
                \taggedAbs{x}{\alpha\impl\beta\impl\alpha}
            \end{prooftree}
        \]
        \item A proof of $\vdash \alpha\impl\beta\impl\alpha\impl\beta$.
        \[
            \begin{prooftree}
                \closedAss{x}{\beta}
                \taggedAbs{}{\alpha\impl\beta}
                \taggedAbs{x}{\beta\impl\alpha\impl\beta}
                \taggedAbs{}{\alpha\impl\beta\impl\alpha\impl\beta}
            \end{prooftree}
        \]
        \item A proof of $\alpha,\beta \vdash \alpha\conj\beta$.
        \[
            \begin{prooftree}
                \hypo{\alpha}
                \hypo{\beta}
                \infer2[\rnpair]{\alpha\conj\beta}
            \end{prooftree}
        \]
        \item A proof of $\gamma \vdash \alpha\impl(\beta\disj\alpha)$. 
        \[
            \begin{prooftree}
                    \closedAss{x}{\alpha}
                    \hypo{\gamma}
                    \infer2[\rnpair]{\alpha\conj\gamma}
                    \infer1[\rnleft{}]{(\alpha\conj\gamma)\disj\beta}
                        
                    \closedAss{y}{\alpha\conj\gamma}
                    \infer1[\rnfst]{\alpha}
                    \infer1[\rnright{}]{\beta\disj\alpha}
                    
                    \closedAss{z}{\beta}
                    \infer1[\rnleft{}]{\beta\disj\alpha}
                \taggedCase{y,z}{\beta\disj\alpha}
                \taggedAbs{x}{\alpha\impl(\beta\disj\alpha)}
            \end{prooftree}
        \]
    \end{enumerate}
\end{example}

We do not formalize this translation from proof terms to proof trees, trusting that the above examples present a reasonable picture. We only note that in the $\rnabs{}$ and $\rncase$ rules, where specific assumptions are \defemph{closed} or \defemph{discharged}, we indicate it by adding a horizontal line on top of the assumptions, along with a \defemph{tag} (a lowercase letter, typically chosen from $a,b,c,x,y,z,w$). We place the same tag to the left of the $\rnabs{}$ or $\rncase$ rule in question. Assumptions without any line on top of them are \defemph{open} (or \defemph{undischarged}). We say that a tree proves $X \vdash \alpha$ if $\alpha$ is the conclusion of the root, and $X$ is the set of undischarged assumptions. 
   
\section{Weak normalization -- key ideas and challenges}\label{sec:challenge}
A very desirable property of proofs is the following -- all formulas occurring in a proof come from a restricted set. More specifically, every $\alpha \in \formsof(\pi)$ is either a subformula of $\concof(\pi)$ or a subformula of some $\beta \in \assof(\pi)$. If every provable judgement has a proof satisfying the \defemph{subformula property} above, then proof search is simplified immensely. To search for a proof of a judgement $X \vdash \alpha$, we only need to consider proofs whose formulas are from $\sfof(X \cup \{\alpha\})$. 

Not all proofs have the subformula property. For instance, the following proof of $\alpha \vdash \alpha\conj\alpha$
\[
    \begin{prooftree}
        \closedAss{x}{\alpha}
        \closedAss{x}{\alpha}
        \infer2[\rnpair]{\alpha\conj\alpha}
        \taggedAbs{x}{\alpha\impl(\alpha\conj\alpha)}
        \hypo{\alpha}
        \infer2[\rnapp]{\alpha\conj\alpha}
    \end{prooftree}
\]
has an occurrence of $\alpha\impl(\alpha\conj\alpha)$ which is neither a subformula of the conclusion nor of any assumption. The same proof in our official syntax is $\rnapp(\rnabs{\alpha}(\rnpair(\rnax{\alpha},\rnax{\alpha})), \rnax{\alpha})$, and it features the pattern $\rnapp(\rnabs{\alpha}(\mu), \nu)$. Such undesirable patterns are called \defemph{cuts} or \defemph{redexes}, and they typically lead to a violation of the subformula property. Proofs without these undesirable patterns are termed \defemph{cut-free} or \defemph{normal}.  

If we can convert every proof to a normal proof with the same conclusion and assumptions, then it would be very useful indeed for automatic proof search, and also for ensuring other desirable structural properties. The fact that any proof can be converted to an equivalent normal proof is referred to as \defemph{Normalization}. Normalization is typically established by defining a set of transformations that can be performed repeatedly on proofs to eliminate cuts. For example, a cut of the form 
\[
    \begin{prooftree}
        \closedAss{x}{\alpha}
        \ellipsis{$\mu$}{\beta}
        \taggedAbs{x}{\alpha\impl\beta}
        \hypo{}
        \ellipsis{$\nu$}{\alpha}
        \infer2[\rnapp]{\beta}
    \end{prooftree}
\]
(with the subproof $\mu$ containing multiple instance of the discharged assumption $\alpha$) can be reduced to the following proof
\[
    \begin{prooftree}
        \hypo{}
        \ellipsis{$\nu$}{\alpha}
        \ellipsis{$\mu$}{\beta}
    \end{prooftree}
\]
obtained by replacing each occurrence of $\rnax{\alpha}$ in $\mu$ with a copy of $\nu$. 

This reduction has the potential to increase the size of the proof, but we identify some other measure that decreases. An approach which works for $\disj$-free fragments is to define the degree of a cut to be the size of the largest formula that features in the cut ($\alpha\impl\beta$ in the above case) and ensure that we apply the above reduction only in cases where all cuts in $\mu$ and $\nu$ have degree strictly less than $\formsize{\alpha\impl\beta}$. We then prove that the proof that results after reduction has maximum degree less than $\formsize{\alpha\impl\beta}$. But this breaks down in the presence of disjunction and the $\rncase$ rule. Consider the following proof $\pi$, for example. 
\[
    \begin{prooftree}
        \hypo{}
        \ellipsis{$\chi$}{\alpha\disj\beta}
        \closedAss{x}{\alpha}
        \ellipsis{$\mu$}{\gamma\impl\delta}
        \closedAss{y}{\beta}
        \ellipsis{$\nu$}{\gamma\impl\delta}
        \taggedCase{x,y}{\gamma\impl\delta}
        \hypo{}
        \ellipsis{$\rho$}{\gamma}
        \infer2[\rnapp]{\delta}
    \end{prooftree}
\]
It is possible that either $\mu$ or $\nu$ itself, or one of their subproofs, introduces $\gamma\impl\delta$ via the $\rnabs{\gamma}$ rule. This is eliminated at the root of $\pi$, potentially a long distance away from the rule that introduced the implication. The only route to normalizing this proof is to bring the $\rnapp$ closer to the source of the introduction, via the following reduction. 
\[
    \begin{prooftree}
        \hypo{}
        \ellipsis{$\chi$}{\alpha\disj\beta}
        \closedAss{x}{\alpha}
        \ellipsis{$\mu$}{\gamma\impl\delta}
        \hypo{}
        \ellipsis{$\rho$}{\gamma}
        \infer2[\rnapp]{\delta}
        \closedAss{y}{\beta}
        \ellipsis{$\nu$}{\gamma\impl\delta}
        \hypo{}
        \ellipsis{$\rho$}{\gamma}
        \infer2[\rnapp]{\delta}
        \taggedCase{x,y}{\delta}
    \end{prooftree}
\]
The reduced proof is likely to have larger size than $\pi$, because $\rho$ is repeated twice. But something does reduce -- in $\pi$ there are two sequences of proofs -- $\mu_{1}, \ldots, \mu_{k}, \mu, \pi$ and $\nu_{1}, \ldots, \nu_{\ell}, \nu, \pi$, with each sequence maximal s.t.\ all $\mu_{i}$'s and $\nu_{j}$'s have conclusion $\gamma\impl\delta$, each $\mu_{i}$ is an immediate subproof of $\mu_{i+1}$ and $\mu_{k}$ an immediate subproof of $\mu$, with similar conditions being satisfied by the $\nu_{j}$'s. Such sequences of proofs are \defemph{cut segments} of degree $\formsize{\gamma\impl\delta}$. The sum of the length of cut segments of this degree in the reduced proof can be shown to be smaller than the corresponding sum in $\pi$, under appropriate conditions. 

\defemph{This brings us to our first technical contribution} -- rather than consider cut segments, we consider $\pi$ of the above form as a cut of degree $\formsize{\gamma\impl\delta}$ and \defemph{weight} $\sizeof(\chi) + \sizeof(\mu) + \sizeof(\nu) + 1$. Note that the weight is the size of the main subproof of $\pi$. For cuts whose main subproof ends in $\rncase$, we take the weight to be the size of their main subproof, while for other cuts, we take the weight to be $1$. We can show that in the reduced proof, the weight decreases under appropriate conditions.  

A further subtlety is revealed by the following proof $\pi$.
\[
    \begin{prooftree}
        \hypo{}
        \ellipsis{$\chi$}{\alpha\disj\beta}
        \hypo{}
        \ellipsis{$\mu$}{\gamma\disj\delta}
        \hypo{}
        \ellipsis{$\nu$}{\gamma\disj\delta}
        \infer3[\rncase]{\gamma\disj\delta}
        \hypo{}
        \ellipsis{$\rho$}{\phi\impl\psi}
        \hypo{}
        \ellipsis{$\upsilon$}{\phi\impl\psi}
        \infer3[\rncase]{\phi\impl\psi}
        \hypo{}
        \ellipsis{$\varpi$}{\phi}
        \infer2[\rnapp]{\psi}
    \end{prooftree}
\]
Suppose $\pi_{1}$ denotes the main subproof of $\pi$, and suppose that $\formsize{\gamma\disj\delta} = \formsize{\phi\impl\psi}$. Then reducing $\pi_{1}$ would lead to an increase in weight of the main suproof of the cut $\pi$. So we must actually reduce $\pi$, even though its subproof $\pi_{1}$ has same degree as $\pi$. So we have moved away from the strategy of reducing the rightmost innermost cut of maximum degree, which works for the fragment without disjunction. \defemph{This brings us to our second technical contribution} -- we identify precise conditions under which we contract at the root of a cut rather than reduce a subproof, even if some subproofs have the same degree as the cut itself. Once we identify these rules, we can phrase the entire process as a deterministic algorithm for one-step reduction. 

We hope that the preceding discussion is helpful in motivating the various definitions and particular choices made in the technical development of the normalization proof presented next.

\section{Normalization}\label{sec:njNorm}
We assign a \defemph{cut-rank} for each proof, which is a pair of natural numbers. Cut-ranks are compared by lexicographic ordering on pairs, which is a well-ordering. We define an associative and commutative operation $\rankadd$ on pairs of natural numbers, as follows.
{
    \setlength{\abovedisplayskip}{1pt}
    \setlength{\belowdisplayskip}{1pt}
    \[
        (d,n) \rankadd (d',n') \coloneq 
        \begin{cases}
            (d',n') & \text{if }d < d' \\ 
            (d, n+n') & \text{if }d = d' \\
            (d,n) & \text{if }d > d'
        \end{cases}
    \]
}        

\begin{definition}[Cuts and cut-rank]
    The following table lists the patterns that define a cut $\mu$, along with its \defemph{cut-formula} (denoted $\cutform(\mu)$), \defemph{degree} (denoted $\degree(\mu)$), and \defemph{weight} (denoted $\cutwt(\mu)$).
    \begin{longtable}{C|C|C|C|C}
        \mu & \text{Preonditions} & \text{Cut-formula} & \text{Degree} & \text{Weight} \\
        \hline 
        \hline
        \rnemp{\bot}(\nu) & \text{None} & \bot & 1 & 1 \\ 
        \hline
        \multirow{2}{*}{$\rnrule(\nu, \ldots)$} & \rnrule \in \{\rnapp, \rnfst, \rnsnd, \rncase\} & \multirow{2}{*}{$\concof(\nu)$} & \multirow{2}{*}{$\formsize{\concof(\nu)}$} & \multirow{2}{*}{$1$} \\
        & \lrof(\nu) \in \{\rnemp{}, \rnabs{}, \rnpair, \rnleft{}, \rnright{}\} & & & \\
        \hline        
        \multirow{2}{*}{$\rnrule(\nu, \ldots)$} & \rnrule \in \{\rnapp, \rnfst, \rnsnd, \rncase\} & \multirow{2}{*}{$\concof(\nu)$} & \multirow{2}{*}{$\formsize{\concof(\nu)}$} & \multirow{2}{*}{$\sizeof(\nu)$} \\
        & \lrof(\nu)  = \rncase & & & \\
    \end{longtable}
    If $\mu$ is not a cut, we still define $\degree(\mu)$ and $\cutwt(\mu)$ to be $0$. The \defemph{cut-rank} of a proof $\pi$, denoted $\cutrank(\pi)$, is defined to be $(0,0)$ when $\pi = \rnax{}$, and $(\degree(\pi),\cutwt(\pi)) \rankadd \cutrank(\pi_{1}) \rankadd \cdots \rankadd \cutrank(\pi_{k})$ when $\pi = \rnrule(\pi_{1}, \ldots, \pi_{k})$. A proof $\pi$ is \defemph{normal} if $\cutrank(\pi) = (0,0)$. Note that $\pi$ is normal iff none of its subproofs is a cut.
\end{definition}
We next present some key transformations on proofs, and prove that they decrease the cut-rank (under appropriate conditions). If $\pi$ and $\rho$ are proofs with $\assof(\pi) = X, \assof(\rho) = Y$ and $\concof(\rho) = \alpha$, we can \defemph{graft} the proof $\rho$ on top of $\pi$. Formally, we define $\graft{\pi}{\alpha}{\rho}$ to be the proof obtained by replacing each subproof $\rnax{\alpha}$ in $\pi$ with $\rho$. It is easy to see that if $\alpha\notin{X}$, then $\graft{\pi}{\alpha}{\rho} = \pi$, and also that $\assof(\graft{\pi}{\alpha}{\rho}) \subseteq Y \cup (X-\alpha)$ and $\concof(\graft{\pi}{\alpha}{\rho}) = \concof(\pi)$.

\begin{lemma}[Cut-rank of a graft]\label{lem:graft-cr}
    Suppose $\pi$ and $\rho$ are proofs with $\concof(\rho) = \alpha$, $\cutrank(\pi), \cutrank(\rho) < (d,0)$, and $\formsize{\alpha} < d$. Then $\cutrank(\graft{\pi}{\alpha}{\rho}) < (d,0)$.
\end{lemma}
\begin{proof}
    We prove this by induction on the structure of $\pi$. We denote $\graft{\pi}{\alpha}{\rho}$ as $\overline{\pi}$ for notational ease. 
    \begin{itemize}[nosep]
        \item If $\pi = \rnax{\alpha}$, then $\overline{\pi} = \rho$, and it has cut-rank $< (d,0)$ by assumption. 
        \item If $\pi = \rnrule(\pi_{1}, \ldots, \pi_{k})$, then $\overline{\pi} = \rnrule(\overline{\pi_{1}}, \ldots, \overline{\pi_{k}})$. By IH, $\cutrank(\overline{\pi_{i}}) < (d,0)$ for each $i \le k$. Now we consider two cases.
        
        If $\lrof(\pi_{1}) = \lrof(\overline{\pi_{1}})$, then $\pi$ is a cut iff $\overline{\pi}$ is a cut, and they both have the same cut-formula, and hence the same degree. By assumption, $\degree(\pi) < d$, so we have also have $\degree(\overline{\pi}) < d$.  
        
        If $\lrof(\pi_{1}) \neq \lrof(\overline{\pi_{1}})$, then it has to be that $\pi_{1} = \rnax{\alpha}$ and $\overline{\pi_{1}} = \rho$. If $\overline{\pi}$ is a cut, the cut-formula is $\alpha$ and its degree is $< d$, by assumption.
    \end{itemize} 
    So we see that $\cutrank(\graft{\pi}{\alpha}{\rho}) < (d,0)$. 
\end{proof}

\begin{definition}[Contractum]\label{def:ctm}
    For any cut $\pi$, its \defemph{contractum} $\contractum(\pi)$ is defined as in the table below. 
    \begin{longtable}{l|CC|C}
        & \pi & & \contractum(\pi) \\ 
        \hline
        \rownum. & \rnapp(\rnabs{\alpha}(\mu), \nu) & & \graft{\mu}{\alpha}{\nu} \\ 
        \rownum. & \rnfst(\rnpair(\mu, \nu)) & & \mu \\ 
        \rownum. & \rnsnd(\rnpair(\mu, \nu)) & & \nu \\ 
        \rownum. & \rncase(\rnleft{\beta}(\chi), \mu, \nu) & (\concof(\chi) = \alpha) & \graft{\mu}{\alpha}{\chi} \\ 
        \rownum. & \rncase(\rnright{\alpha}(\chi), \mu, \nu) & (\concof(\chi) = \beta) & \graft{\nu}{\beta}{\chi} \\ 
        \rownum. & \rnemp{\bot}(\mu) & & \mu \\ 
        \rownum. & \rnapp(\rnemp{\alpha\impl\beta}(\mu), \nu) & & \rnemp{\beta}(\mu) \\ 
        \rownum. & \rnfst(\rnemp{\alpha\conj\beta}(\mu)) & & \rnemp{\alpha}(\mu) \\ 
        \rownum. & \rnsnd(\rnemp{\alpha\conj\beta}(\mu)) & & \rnemp{\beta}(\mu) \\ 
        \rownum. & \rncase(\rnemp{\alpha\disj\beta}(\chi), \mu, \nu) & (\concof(\pi) = \delta) & \rnemp{\delta}(\chi) \\ 
        \rownum. & \rnapp(\rncase(\chi, \mu, \nu), \rho) & & \rncase(\chi, \rnapp(\mu, \rho), \rnapp(\nu, \rho)) \\ 
        \rownum. & \rnfst(\rncase(\chi, \mu, \nu)) & & \rncase(\chi, \rnfst(\mu), \rnfst(\nu)) \\ 
        \rownum. & \rnsnd(\rncase(\chi, \mu, \nu)) & & \rncase(\chi, \rnsnd(\mu), \rnsnd(\nu)) \\ 
        \rownum. & \rncase(\rncase(\chi, \mu, \nu), \rho, \upsilon) & & \rncase(\chi, \rncase(\mu, \rho, \upsilon), \rncase(\nu, \rho, \upsilon)) \\ 
    \end{longtable}
\end{definition}

\begin{lemma}\label{lem:ass-ctm}
    For any cut $\pi$, $\assof(\contractum(\pi)) \subseteq \assof(\pi)$ and $\concof(\contractum(\pi)) = \concof(\pi)$. Furthermore, either $\formsize{\concof(\contractum(\pi))} < \degree(\pi)$ or $\lrof(\pi) = \lrof(\contractum(\pi))$ or $\lrof(\pi) \in \{\rnemp{\bot}, \rncase\}$.  
\end{lemma}
\begin{proof}
    A routine case analysis of the structure of $\pi$.
\end{proof}

\begin{lemma}[Cut-rank of a contractum]\label{lem:cr-red-ctm}
    Suppose $\pi$ is a cut of degree $d$, and all its ancilliary subproofs have cut-rank $< (d,0)$. Suppose further that $\lrof(\upsilon) = \rncase$ or $\cutrank(\upsilon) < (d,0)$, where $\upsilon$ is the main subproof of $\pi$. Then $\cutrank(\contractum(\pi)) < \cutrank(\pi)$. 
\end{lemma}
\begin{proof}
    Suppose $\cutrank(\pi) = (d,n)$. The proof is by analysing the last rule of $\pi$ (and that of its main subproof). We consider the various cases below. For ease of notation, we use $\overline{\pi}$ to refer to $\contractum(\pi)$. 
    \begin{itemize}
        \item Suppose $\pi = \rnapp(\rnabs{\alpha}(\mu),\nu)$ with $\concof(\mu) = \beta$, $\concof(\nu) = \alpha$, and $d = \formsize{\alpha\impl\beta}$. Then we have $\overline{\pi} = \graft{\mu}{\alpha}{\nu}$. We also have $\cutrank(\mu), \cutrank(\nu) < (d,0)$ by assumption, and $\formsize{\alpha} < d$. So by Lemma~\ref{lem:graft-cr}, $\cutrank(\overline{\pi}) < (d,0)$. 
        \item Suppose $\pi = \rnfst(\rnpair(\mu,\nu))$ with $\concof(\mu) = \alpha$, $\concof(\nu) = \beta$ and $d = \formsize{\alpha\conj\beta}$. Then we have $\overline{\pi} = \mu$, which has cut-rank $< (d,0)$, by assumption.  
        \item Suppose $\pi = \rnsnd(\rnpair(\mu,\nu))$. This is similar to the above case.
        \item Suppose $\pi = \rncase(\rnleft{\beta}(\chi),\mu,\nu)$ with $\concof(\chi) = \alpha$, $\concof(\mu) = \concof(\nu) = \delta$ and $d = \formsize{\alpha\disj\beta}$. Then we have $\overline{\pi} = \graft{\mu}{\alpha}{\chi}$. By assumption, we have $\cutrank(\chi), \cutrank(\mu), \cutrank(\nu) < (d,0)$. It follows  by Lemma~\ref{lem:graft-cr} that $\cutrank(\overline{\pi}) < (d,0)$, since $\formsize{\alpha} < d$. 
        \item Suppose $\pi = \rncase(\rnright{\alpha}(\chi),\mu,\nu)$. This is similar to the above case.
        \item Suppose $\pi = \rnemp{\bot}(\mu)$. Then $\overline{\pi} = \mu$ which has cut-rank $< (d,0)$ by assumption.
        \item Suppose the main subproof of $\pi$ has last rule $\rnemp{}$. All cases are similar, so we consider the case when $\pi = \rnapp(\rnemp{\alpha\impl\beta}(\mu), \nu)$, with $\concof(\mu) = \bot$, $\concof(\nu) = \alpha$ and $d = \formsize{\alpha\impl\beta}$. Then we have $\overline{\pi} = \rnemp{\beta}(\mu)$. We have $\cutrank(\mu) < (d,0)$ by assumption. If $\overline{\pi}$ itself is a cut, then it has to be that $\beta = \bot$, in which case $\degree(\overline{\pi}) = \formsize{\bot} < d$, so $\cutrank(\overline{\pi}) < (d,0)$. 
        \item Suppose the main subproof of $\pi$ has last rule $\rncase$. All cases are similar, so we consider the case when $\pi = \rncase(\rho, \pi_{1}, \pi_{2})$, with $\rho = \rncase(\chi, \mu, \nu)$. Then we have $\concof(\pi) = \concof(\pi_{1}) = \concof(\pi_{2}) = \phi$, $\concof(\rho) = \concof(\mu) = \concof(\nu) = \alpha\disj\beta$, and $\concof(\chi) = \gamma\disj\delta$ for some $\alpha, \beta, \gamma, \delta, \phi$. Further, we have $\degree(\pi) = d = \formsize{\alpha\disj\beta}$, and $\cutwt(\pi) = \sizeof(\rho) = 1 + \sizeof(\chi) + \sizeof(\mu) + \sizeof(\nu)$. By definition of $\contractum$, we have $\overline{\pi} = \rncase(\chi, \overline{\mu}, \overline{\nu})$ where $\overline{\mu} = \rncase(\mu, \pi_{1}, \pi_{2})$ and $\overline{\nu} = \rncase(\nu, \pi_{1}, \pi_{2})$. We have $\degree(\rho) \le d$ and $\cutrank(\pi_{1}), \cutrank(\pi_{2}) < (d,0)$ by assumption. In forming $\overline{\pi}$, the cuts inside $\chi$, $\mu$ and $\nu$ are neither duplicated nor modified. The cuts inside $\pi_{1}$ and $\pi_{2}$ are duplicated in $\overline{\pi}$ but they are of degree $< d$. It is possible that three new cuts have been introduced -- $\overline{\mu}$, $\overline{\nu}$, and $\overline{\pi}$ itself. The first two of these possible cuts have degree $\formsize{\alpha\disj\beta} = d$, and the sum of their weights is $\le \sizeof(\mu) + \sizeof(\nu)$. If $\overline{\pi}$ is a cut, then it is easy to see that $\rho$ itself is a cut in $\pi$, and its degree is $\le d$. Furthermore, $\cutwt(\overline{\pi}) = \sizeof(\chi)$. Thus we see that the sum of the weights of all the possible new cuts of degree $d$ is strictly less than the weight of the cut $\pi$. Thus $\cutrank(\overline{\pi}) < \cutrank(\pi)$.  
        \qedhere
    \end{itemize}
\end{proof}

\begin{definition}[One-step reduction]\label{def:one-step-red}
    For proofs $\pi$ and $\overline{\pi}$ with $\cutrank(\pi) = (d,n)$, we say that $\pi \oneStep \overline{\pi}$ iff one of the conditions in the table below holds:
    \setcounter{rownum}{0}
    \begin{longtable}{l|C|C|C}
        & \pi & \text{Preconditions} & \overline{\pi} \\ 
        \hline 
        \hline
        \multirow{2}{*}{\rownum.} & \multirow{2}{*}{$\rnrule(\pi_{1}, \ldots, \pi_{i}, \ldots, \pi_{k})$} & \degree(\pi) = d,\quad \cutrank(\pi_{2}), \ldots, \cutrank(\pi_{k}) < (d,0), & \multirow{2}{*}{$\contractum(\pi)$} \\ 
        & & \lrof(\pi_{1}) = \rncase \text{ or }\cutrank(\pi_{1}) < (d,0) & \\   
        \hline  
        \multirow{2}{*}{\rownum.} & \multirow{2}{*}{$\rnrule(\pi_{1}, \ldots, \pi_{k})$} & \lrof(\pi_{1}) \neq \rncase\text{ or }\degree(\pi) < d, & \multirow{2}{*}{$\rnrule(\overline{\pi_{1}}, \ldots, \ \pi_{k})$} \\  
        & & \cutrank(\pi_{1}) \ge (d,0),\quad \pi_{1} \oneStep \overline{\pi_{1}} & \\ 
        \hline  
        \rownum. & \rnrule(\pi_{1}, \ldots, \pi_{i}, \ldots, \pi_{k}) & i > 1,\quad \cutrank(\pi_{i}) \ge (d,0),\quad \pi_{i} \oneStep \overline{\pi_{i}} & \rnrule(\pi_{1}, \ldots, \overline{\pi_{i}}, \ldots, \pi_{k}) \\ 
    \end{longtable}
\end{definition}

\begin{lemma}\label{lem:red-exists}
    Suppose $\pi$ is a non-normal proof. Then there exists $\overline{\pi}$ such that $\pi \oneStep \overline{\pi}$. 
\end{lemma}
\begin{proof}
    We assume that the claim is true for all proper subproofs of $\pi$ and prove it for $\pi$ itself. Suppose $\pi = \rnrule(\pi_{1}, \ldots, \pi_{k})$ and $\cutrank(\pi) = (d,n)$. 
    \begin{itemize}[nosep]
        \item Suppose rule 3 applies. Then $\pi_{i}$ is not normal, and by IH, there exists $\overline{\pi_{i}}$ such that $\pi_{i} \oneStep \overline{\pi_{i}}$, and we can choose $\overline{\pi}$ to be $\rnrule(\pi_{1}, \ldots, \overline{\pi_{i}}, \ldots \pi_{k})$.  
        \item Suppose rule 2 applies. Then $\pi_{1}$ is not normal, and by IH, there exists $\overline{\pi_{1}}$ such that $\pi_{1} \oneStep \overline{\pi_{1}}$, and we can choose $\overline{\pi}$ to be $\rnrule(\overline{\pi_{1}}, \ldots, \pi_{k})$.   
        \item Suppose neither rule 2 nor rule 3 apply. Then $\cutrank(\pi_{i}) < (d,0)$ for $2 \le i \le k$, and either $\cutrank(\pi_{1}) < (d,0)$, or $\lrof(\pi_{1}) = \rncase$ and $\degree(\pi) \ge d$. But in the former case, $\degree(\pi) = d$, as otherwise the maximum degree of any cut in $\pi$ (which is the maximum among $\degree(\pi)$ and the maximum degrees of all the $\pi_{i}$'s) would be less than $d$, contradicting the cut-rank of $\pi$. So in either case, rule 1 applies, and we choose $\overline{\pi}$ to be $\contractum(\pi)$. 
        \qedhere 
    \end{itemize}
\end{proof}

\begin{lemma}[Cut-rank after a reduction]\label{lem:cr-red-red}
    If $\pi \oneStep \overline{\pi}$, then $\assof(\overline{\pi}) \subseteq \assof(\pi)$, $\concof(\overline{\pi}) = \concof(\pi)$ and $\cutrank(\overline{\pi}) < \cutrank(\pi)$. 
\end{lemma}
\begin{proof}
    The statement about $\assof(\overline{\pi})$ and $\concof(\overline{\pi})$ is proved by a simple induction on the structure of $\pi$, using the corresponding claim about $\contractum$ in Lemma~\ref{lem:ass-ctm}. We prove the claim about cut-rank by induction on the structure of proofs. So suppose the statement is true for all smaller proofs, and consider $\pi = \rnrule(\pi_{1}, \ldots, \pi_{k})$ with $\cutrank(\pi) = (d,n)$ and $d > 0$. Suppose also that $\cutrank(\pi_{j}) = (d_{j}, n_{j})$ for each $j \le k$. We do a case analysis on the rule (in the table of Definition~\ref{def:one-step-red}) that yields $\pi \oneStep \overline{\pi}$. 
    \begin{itemize}
        \item Suppose rule 1 applies. Then $\overline{\pi} = \contractum(\pi)$, and it follows from Lemma~\ref{lem:cr-red-ctm} that $\cutrank(\overline{\pi}) < \cutrank(\pi)$. 
        \item Suppose rule 2 applies. Then $\cutrank(\pi_{1}) \ge (d,0)$, and $\lrof(\pi_{1}) \neq \rncase$ or $\degree(\pi) < d$. For concreteness, assume that $k = 3$, $d_{1} = d$, $\pi_{1} \oneStep \overline{\pi_{1}}$, and $\overline{\pi} = \rnrule(\overline{\pi_{1}}, \pi_{2},\pi_{3})$. Suppose $\cutrank(\overline{\pi_{1}}) = (d', n')$. By IH, we have $(d',n') < (d_{1}, n_{1})$. Let $(e,m) = (\degree(\pi), \cutwt(\pi))$ and $(e',m') = (\degree(\overline{\pi}), \cutwt(\overline{\pi}))$. We consider three cases now.
        \begin{itemize}[nosep]
            \item Suppose $\lrof(\pi_{1}) \neq \rncase$ and rule 2 or 3 for $\oneStep$ applies to $\pi_{1}$. In this case, we have $\lrof(\overline{\pi_{1}}) = \lrof(\pi_{1})$. This means that $e = e'$ and $m = m' = 1$, so $(e',m') \le (e,m)$. 
            \item Suppose $\lrof(\pi_{1}) \neq \rncase$ and rule 1 for $\oneStep$ applies to $\pi_{1}$. So $\pi_{1}$ is a cut of degree $d$, and by Lemma~\ref{lem:ass-ctm} we have $\formsize{\concof(\overline{\pi_{1}})} < d$ or $\lrof(\pi_{1}) \in \{\rnemp{\bot}, \rncase\}$. But since $\lrof(\pi_{1}) \neq \rncase$, we have $\formsize{\concof(\overline{\pi_{1}})} < d$ or $\lrof(\pi_{1}) = \rnemp{\bot}$. It follows that $(e',m') < (d,0)$ or $(e',m') = (e,m) = (1,1)$.  
            \item Suppose $\lrof(\pi_{1}) = \rncase$. Then $\degree(\pi) < d$ (since the conditions for rule 2 are satisfied). Now if $\overline{\pi}$ is a cut, so is $\pi$, and both degree $< d$. So $(e',m') < (d,0)$. 
        \end{itemize}
        So we see that in all cases, $(e',m') < (d,0)$ or $(e',m') \le (e,m)$, and we further have that $d_{1} = d = \max(e,d_{1},d_{2},d_{3})$. Therefore the following inequality holds, as can be checked easily. 
        {
            \setlength{\abovedisplayskip}{2pt}
            \setlength{\belowdisplayskip}{2pt}
            \[
                (e',m') \rankadd (d',n') \rankadd (d_{2},n_{2}) \rankadd (d_{3},n_{3}) < (e,m) \rankadd (d_{1},n_{1}) \rankadd (d_{2},n_{2}) \rankadd (d_{3},n_{3}).
            \]
        }        
        But this just means that $\cutrank(\overline{\pi}) < \cutrank(\pi)$. 
        \item Suppose rule 3 applies. Then there is $i > 1$ s.t.\ $\cutrank(\pi_{i}) \ge (d,0)$ and $\pi_{i} \oneStep \overline{\pi_{i}}$. For concreteness, assume that $k = 3$, $i = 2$, $d_{2} = d$, and $\overline{\pi} = \rnrule(\pi_{1}, \overline{\pi_{2}}, \pi_{3})$. Suppose $\cutrank(\overline{\pi_{2}}) = (d', n')$. By IH, we have $(d',n') < (d_{2}, n_{2})$. Since $\pi_{1}$ is the main subproof of both $\pi$ and $\overline{\pi}$, we see that $\overline{\pi}$ is a cut iff $\pi$ is a cut with the same degree and weight. Letting $(e,m) = (\degree(\pi), \cutwt(\pi))$, we have that $d_{2} = \max(e, d_{1},d_{2},d_{3})$, and hence 
        {
            \setlength{\abovedisplayskip}{2pt}
            \setlength{\belowdisplayskip}{2pt}
            \[
                (e,m) \rankadd (d_{1},n_{1}) \rankadd (d',n') \rankadd (d_{3},n_{3}) < (e,m) \rankadd (d_{1},n_{1}) \rankadd (d_{2},n_{2}) \rankadd (d_{3},n_{3}),
            \]
        }        
        which means that $\cutrank(\overline{\pi}) < \cutrank(\pi)$.
        \qedhere
    \end{itemize}
\end{proof}

\begin{theorem}[Normalization]\label{thm:normal}
    If\, $X \intDerives \alpha$, there is a normal proof of\, $Y \vdash \alpha$ for some $Y \subseteq X$. 
\end{theorem}
\begin{proof}
    We show that for any proof $\pi$, there is a normal proof $\pi^{0}$ with a subset of assumptions and same conclusion as $\pi$. Suppose the claim is true for all proofs of smaller cut-rank. If $\pi$ itself is normal, then we take $\pi^{0} \coloneq \pi$. Otherwise, by Lemma~\ref{lem:red-exists}, there is a $\overline{\pi}$ with a subset of assumptions and same conclusion as $\pi$, and such that $\cutrank(\overline{\pi}) < \cutrank(\pi)$. By induction hypothesis, there is normal proof $\pi^{0}$ with a subset of assumptions and same conclusion as $\overline{\pi}$, and we are done. 
\end{proof}

A proof $\pi$ is \defemph{neutral} if $\lrof(\pi) \in \{\rnax{}, \rnapp, \rnfst, \rnsnd\}$, or if $\concof(\pi) = \bot$.   
\begin{theorem}[Subformula Property]
    If $\pi$ is a normal proof, then $\formsof(\pi) \subseteq \sfof(\assof(\pi) \cup \{\concof(\pi)\})$. Furthermore, if $\pi$ is neutral, then $\formsof(\pi) \subseteq \sfof(\assof(\pi))$. 
\end{theorem}
\begin{proof}
    The proof is by induction on the structure of proofs. So suppose the claim is true for all proper subproofs $\mu$ of $\pi$ (which are all normal), and consider the various cases for $\pi$. 
    \begin{description}[nosep]
        \item[$\pi = \rnax{\alpha}$] Then $\formsof(\pi) = \{\alpha\} \subseteq \sfof(\{\alpha\})$, as desired. 
        \item[$\pi = \rnemp{\alpha}(\mu)$] Let $\assof(\pi) = \assof(\mu) = X$. Note that $\concof(\mu) = \bot$, so by IH, we have $\formsof(\mu) \subseteq \sfof(X)$. It follows that $\formsof(\pi) = \formsof(\mu) \cup \{\alpha\} \subseteq \sfof(X \cup \{\alpha\})$, as desired. It is not possible that $\alpha = \bot$, since the cut $\rnemp{\bot}(\mu)$ cannot occut in a normal proof. 
        \item[$\pi = \rnabs{\alpha}(\mu)$] Let $\assof(\mu) = X$ and $\concof(\mu) = \beta$. Then $\assof(\pi) = X-\alpha$ and $\concof(\pi) = \alpha\impl\beta$. By IH, we have $\formsof(\mu) \subseteq \sfof(X\cup\{\beta\})$. But $\sfof(X\cup\{\beta\}) \subseteq \sfof((X-\alpha)\cup\{\alpha\impl\beta\})$. It follows that $\formsof(\pi) \subseteq \{\alpha\impl\beta\} \cup \formsof(\mu) \subseteq \sfof((X-\alpha) \cup \{\alpha\impl\beta\})$, as desired. 
        \item[$\pi = \rnapp(\mu,\nu)$] Let $\assof(\mu) = X$ and $\assof(\nu) = Y$. Let $\concof(\mu) = \alpha\impl\beta$. Then $\assof(\pi) = X\cup{Y}$, $\concof(\nu) = \alpha$ and $\concof(\pi) = \beta$. Since $\pi$ is normal, $\mu$ is neutral, so by IH, $\formsof(\mu) \subseteq \sfof(X)$. In particular, $\{\alpha, \beta, \alpha\impl\beta\} \subseteq \sfof(X)$. By IH, $\formsof(\nu) \subseteq \sfof(Y \cup \{\alpha\}) \subseteq \sfof(X\cup{Y})$. It follows that 
        {
            \setlength{\abovedisplayskip}{2pt}
            \setlength{\belowdisplayskip}{2pt}
            \[
                \formsof(\pi) = \{\beta\} \cup \formsof(\mu) \cup \formsof(\nu) \subseteq \sfof(X \cup Y),
            \]
        }        
        as desired. 
        \item[$\lrof(\pi) \in \{\rnpair, \rnfst, \rnsnd, \rnleft{}, \rnright{}\}$] The argument is similar to one of the above cases. 
        \item[$\pi = \rncase(\chi, \mu, \nu)$] Let $\assof(\chi) = X$, $\concof(\chi) = \alpha\disj\beta, \assof(\mu) = Y, \assof(\nu) = Z$ and $\concof(\pi) = \concof(\mu) = \concof(\nu) = \delta$. Then $\assof(\pi) = X\cup(Y-\alpha)\cup(Z-\beta)$. Since $\pi$ is normal, $\chi$ is neutral. So $\formsof(\chi) \subseteq \sfof(X)$, by IH. In particular, $\{\alpha,\beta,\alpha\disj\beta\} \subseteq \sfof(X)$. Also by IH, 
        {
            \setlength{\abovedisplayskip}{2pt}
            \setlength{\belowdisplayskip}{2pt}
            \[
                \formsof(\mu) \cup \formsof(\nu) \subseteq \sfof(\{\delta\} \cup Y \cup Z) \subseteq \sfof(X \cup (Y-\alpha)\cup (Z-\beta) \cup \{\delta\}). 
            \]
        }        
        It follows that 
                {
            \setlength{\abovedisplayskip}{2pt}
            \setlength{\belowdisplayskip}{2pt}
            \[
                \formsof(\pi) = \{\delta\} \cup \formsof(\chi) \cup \formsof(\mu) \cup \formsof(\nu) \subseteq \sfof(X \cup (Y-\alpha)\cup (Z-\beta) \cup \{\delta\}),
            \]
        }        
        as desired. If $\delta = \bot$, then $\concof(\mu) = \bot$, so by IH, $\formsof(\mu) \subseteq \sfof(Y)$. Since $\bot \in \formsof(\mu)$, we have $\bot \in \sfof(Y)$. If $\bot \notin \sfof(Y-\alpha)$, then $\bot \in \sfof(\alpha) \subseteq \sfof(X)$, so $\bot \in \sfof(X \cup (Y-\alpha)\cup (Z-\beta))$, as desired. 
        \qedhere
    \end{description}
\end{proof}

\section{Lean formalization}~\label{sec:formalization}
We have formalized all the definitions and proofs of the previous sections in Lean. The code repository is hosted at~\cite{NatDed}. Our definitions of proofs, and various parameters associated with a proof, like size, cutrank, are inductive, and the notions of contractum, graft and $\pi \oneStep \overline{\pi}$ are inductive, and so can be transcribed into Lean more or less directly. But our definition of important notions like cut-rank and one-step reduction ($\pi \oneStep \overline{\pi}$) are quite succinct, as we compress many cases under a common pattern using the $\rnrule(\cdots)$ notation. In Lean, we have to elaborate this by writing out all the cases. The main proofs are more or less elaborations of the proofs in Section~\ref{sec:njNorm}, at the cost of repeating many arguments for different explicit cut patterns. For instance, to argue that the cut-rank reduces after contracting a cut whose main subproofs ends with the $\rncase$ rule, we have to consider four different lemmas explicitly, based on whether the last rule of the cut itself is $\rnapp$, $\rnfst$, $\rnsnd$ or $\rncase$. But all the key ideas are already  presented in Section~\ref{sec:njNorm}. 

The proof of Lemma~\ref{lem:red-exists} suggests a way to turn the one-step reduction relation $\oneStep$ into a one-step reduction \emph{function} $\os$. The pseudocode is presented below (we use $\maxdeg(\pi)$ to denote the maximum degree among all cuts in $\pi$). 

\begin{center}
    \begin{algorithmic}
    \Function{\os}{$\pi$}
        \Switch{$\pi$}
            \Case{$\rnax{\alpha}$}
                \quad \Return $\pi$
            \EndCase
            \Case{$\rnemp{\alpha}(\mu)$}
                \If {$\maxdeg(\mu) = \maxdeg(\pi)$}
                    \State \Return $\rnemp{\alpha}(\os(\mu))$
                \Else 
                    \ \Return $\contractum(\pi)$
                \EndIf
            \EndCase 
            \Case{$\rnrule(\mu), \rnrule \in \{\rnfst, \rnsnd\}$}
                \If {$\maxdeg(\mu) = \maxdeg(\pi) \conj (\degree(\pi) < \maxdeg(\pi) \disj \lrof(\mu) \neq \rncase)$}
                    \State \Return $\rnrule(\os(\mu))$ 
                \Else 
                    \ \Return $\contractum(\pi)$
                \EndIf
            \EndCase  
            \Case{$\rnapp(\mu, \nu)$}
                \If{$\maxdeg(\nu) = \maxdeg(\pi)$}
                    \State \Return $\rnapp(\mu, \os(\nu))$ 
                \ElsIf{$\maxdeg(\mu) = \maxdeg(\pi) \conj (\degree(\pi) < \maxdeg(\pi) \disj \lrof(\mu) \neq \rncase)$}
                    \State \Return $\rnapp(\os(\mu), \nu)$
                \Else
                    \ \Return $\contractum(\pi)$ 
                \EndIf
            \EndCase 
            \Case{$\rncase(\chi, \mu, \nu)$}
                \If{$\maxdeg(\nu) = \maxdeg(\pi)$}
                    \State \Return $\rncase(\chi, \mu, \os(\nu))$ 
                \ElsIf{$\maxdeg(\mu) = \maxdeg(\pi)$}
                    \State \Return $\rncase(\chi, \os(\mu), \nu)$ 
                \ElsIf{$\maxdeg(\chi) = \maxdeg(\pi) \conj (\degree(\pi) < \maxdeg(\pi) \disj \lrof(\chi) \neq \rncase)$}
                    \State \Return $\rnapp(\os(\chi), \mu, \nu)$
                \Else
                    \ \Return $\contractum(\pi)$ 
                \EndIf
            \EndCase 
            \Case{$\rnrule(\mu), \rnrule \in \{\rnabs{}, \rnleft{}, \rnright{}\}$}
                \quad \Return $\rnrule(\os(\mu))$
            \EndCase
            \Case{$\rnpair(\mu,\nu)$}
                \If{$\maxdeg(\nu) = \maxdeg(\pi)$}
                    \State \Return $\rnpair(\mu, \os(\nu))$
                \Else
                    \ \Return $\rnpair(\os(\mu), \nu)$
                \EndIf
            \EndCase
        \EndSwitch 
    \EndFunction 
    \end{algorithmic}
\end{center}

The $\os$ function is used to define the normalization function, and both the functions been coded as (terminating) recursive functions in the Lean formalization, and can be executed via the commands \texttt{\#eval oneStep pi} and \texttt{\#eval normalize pi}. 
\begin{algorithmic}
\Function{\mbox{\textsc{normalize}}}{$\pi$}
    \If{$\maxdeg(\pi) = 0$}
        \ \Return $\pi$ 
    \Else 
        \ \Return $\textsc{normalize}(\os(\pi))$
    \EndIf
\EndFunction
\end{algorithmic}

\section{\texorpdfstring{$\classlogic$}{\textsf{NK}}: Natural Deduction for Classical Propositional Logic}~\label{sec:nk}
Classical logic $\classlogic$ has the rule \defemph{reductio ad absurdum} rule $\rnraa{\alpha}$ -- to conclude $\alpha$, assume $\neg\alpha$ and derive a contradiction. We adopt a version of $\rncase$ where the conclusion is always $\bot$. As we will see, $\rnemp{}$ can be simulated by $\rnraa{}$, and so we drop it. The rest of the rules are the same. The key changes in the definition of proofs are given below. 
\begin{definition}[$\classlogic$ proofs] 
    $\classlogic$ proofs are built using the constructors 
    \[
        \rnax{}, \rnraa{}, \rnabs{}, \rnapp, \rnpair, \rnfst, \rnsnd, \rnleft{}, \rnright{}, \rncase
    \] 
    where 
    \begin{itemize}[nosep]
        \item If $\mu \in \Pi_{\bot}$, then $\rnraa{\alpha}(\mu) \in \Pi_{\alpha}$.
        \item If $\chi \in \Pi_{\alpha\disj\beta}$ and $\mu, \nu \in \Pi_{\bot}$, then $\rncase(\chi, \mu, \nu) \in \Pi_{\bot}$. 
    \end{itemize}
    The rest of the constructors behave exactly as in $\intlogic$. The basic definitions of $\concof$, $\assof$, $\formsof$, $\spof$ etc.\ are as before, except for the one change -- $\assof(\rnraa{\alpha}(\mu)) = \assof(\mu) - \neg\alpha$.   
\end{definition} 

The $\rnemp{}$ rule may be safely dropped, since all instances of it in a proof can be replaced by $\rnraa{}$, yielding a proof with possibly fewer assumptions. More importantly, the general $\rncase$ rule with conclusion $\delta$ can be simulated using the restricted $\rncase$ and $\rnraa{\delta}$ as follows:
\[
    \rnraa{\delta}(\rncase(\chi, \rnapp(\rnax{\neg\delta}, \mu), \rnapp(\rnax{\neg\delta}, \nu))).
\] 

Because of the simplified $\rncase$, we can never have an elimination rule whose main subproof ends with $\rncase$. Thus no cut is of the form $\rnrule(\rncase(\dots), \dots)$. This means that we do not have to deal with the complications associated with permutative conversions. We can also define the weight of a cut to be $1$, and can go back to the simple strategy of reducing the rightmost uppermost cut of maximum degree. But there are new complications that arise with the $\rnraa{}$ rule. In the case of $\rnemp{}$, we could simply contract a cut of the form $\rnapp(\rnemp{\alpha\impl\beta}(\mu), \nu)$ to $\rnemp{\beta}(\mu)$. But we cannot do the same with $\rnraa{}$ while still preserving the set of open assumptions of the proof. What is needed is a new set of transformations that operate at the level of assumptions. Consider the following proof $\pi$, with $\mu$ having (possibly several instances of) the assumption $\neg(\alpha\impl\beta)$, which have been discharged while forming $\pi$.
\[
    \begin{prooftree}
        \closedAss{x}{\neg(\alpha\impl\beta)}
        \ellipsis{$\mu$}{\bot}
        \taggedRaa{x}{\alpha\impl\beta}
        \hypo{}
        \ellipsis{$\nu$}{\alpha}
        \infer2[\rnapp]{\beta}
    \end{prooftree}
\]
This is a new form of cut that arises because of the interaction of $\rnraa{}$ with an elimination rule. It would be ideal to transform this to a proof of the form
\[
    \begin{prooftree}
        \closedAss{x}{\neg\beta}
        \ellipsis{$\mu$}{\bot}
        \taggedRaa{x}{\beta}
    \end{prooftree}
\]
where the all undischarged instances of the assumption $\neg(\alpha\impl\beta)$ in $\mu$ have been replaced by $\neg\beta$. To ensure that we still have a well-formed proof, we need to consider the context in which the assumption appears in $\mu$, and we also need to use the proof $\nu$. We present several illustrative cases below, which will motivate the formal definitions that will be presented later. 

Consider a proof of the form 
\[
    \begin{prooftree}
        \hypo{\neg(\alpha\impl\beta)} 
        \hypo{}
        \ellipsis{$\rho$}{\gamma}
        \infer2[\rnpair]{\neg(\alpha\impl\beta) \conj \gamma}
        \ellipsis{$\mu$}{\bot}
    \end{prooftree}
\]
Given a proof $\nu$ with conclusion $\alpha$, we can eliminate the assumption $\neg(\alpha\impl\beta)$ (and replace it with $\neg\beta$) as follows.
\[
    \begin{prooftree}
        \hypo{\neg\beta}
        \closedAss{x}{\alpha\impl\beta}
        \hypo{}
        \ellipsis{$\nu$}{\alpha}
        \infer2[\rnapp]{\beta}
        \infer2[\rnapp]{\bot}
        \taggedAbs{x}{\neg(\alpha\impl\beta)}
        \hypo{}
        \ellipsis{$\rho$}{\gamma}
        \infer2[\rnpair]{\neg(\alpha\impl\beta) \conj \gamma}
        \ellipsis{$\mu$}{\bot}
    \end{prooftree}
\]
Stated precisely using proof terms, the transformation is as follows: 
\[
    \rnax{\neg(\alpha\impl\beta)} \quad \rightsquigarrow \quad \rnabs{\alpha\impl\beta}(\rnapp(\rnax{\neg\beta}, \rnapp(\rnax{\alpha\impl\beta}, \nu))).
\] 
But we need to ensure that we do not introduce new cuts with cut-formula $\neg(\alpha\impl\beta)$ in doing this transformation. Suppose we started with the proof below.  
\[
    \begin{prooftree}
        \hypo{\neg(\alpha\impl\beta)} 
        \hypo{}
        \ellipsis{$\rho$}{\alpha\impl\beta}
        \infer2[\rnpair]{\bot}
        \ellipsis{$\mu$}{\bot}
    \end{prooftree}
\]
If we apply the same transformation as earlier, then we would be introducing a new cut. Instead, we must transform the  proof as follows.
\[
    \begin{prooftree}
        \hypo{\neg\beta}
        \hypo{}
        \ellipsis{$\rho$}{\alpha\impl\beta}
        \hypo{}
        \ellipsis{$\nu$}{\alpha}
        \infer2[\rnapp]{\beta}
        \infer2[\rnapp]{\bot}
        \ellipsis{$\mu$}{\bot}
    \end{prooftree}
\]

There is a further complication -- $\rho$ itself might end with $\rnraa{}$. 
\[
    \begin{prooftree}
        \hypo{\neg(\alpha\impl\beta)} 
        \closedAss{x}{\neg(\alpha\impl\beta)}
        \ellipsis{$\rho$}{\bot}
        \taggedRaa{x}{\alpha\impl\beta}
        \infer2[\rnpair]{\bot}
        \ellipsis{$\mu$}{\bot}
    \end{prooftree}
\]
We must transform this as follows.
\[
   \begin{prooftree}
        \hypo{\neg(\alpha\impl\beta)} 
        \ellipsis{$\rho$}{\bot}
        \ellipsis{$\mu$}{\bot}
    \end{prooftree}
\]
What is happening here? The open assumption $\neg(\alpha\impl\beta)$ has been eliminated, but the \emph{closed} instance of the same assumption has become open now, since we removed an $\rnraa{\alpha\impl\beta}$ rule. To handle this in general, we must first apply the transformation to $\rho$ recursively. This is formalized in the next section.

We presented the $\rnapp$ case as an illustration, but there are similar cuts that happen when main subproof of a cut ends with $\rnraa{}$, and the cut itself ends with $\rnfst$ or $\rnsnd$ or $\rncase$. In those cases too, we follow a similar transformation as in the $\rnapp$ case. There are two further transformations we perform. 
\begin{itemize}
    \item Any proof of the form $\rnraa{\bot}(\mu)$ is converted to $\mu'$, which is obtained by replacing all assumptions of $\neg\bot$ in $\mu$ with $\rnabs{\bot}(\rnax{\bot})$ (taking care to not introduce new cuts on $\neg\bot$). 
    \item Any proof of the form $\rnraa{\neg\alpha}(\mu)$ is converted to $\rnabs{\alpha}(\mu')$, where $\mu'$ is obtained by replacing all assumptions of $\neg\neg\alpha$ in $\mu$ with proofs of $\neg\neg\alpha$ from assumption $\alpha$. This transformation is crucial for obtaining a clean proof of the subformula property.
\end{itemize}

With these motivations in place, we must decide on the degree of cuts involving the $\rnraa{\alpha\impl\beta}$ rule, say. We cannot just take it to be $\formsize{\alpha\impl\beta}$ since the transformations described above might potentially introduce cuts with cut-formula $\alpha\impl\beta$. Since we are discharging assumptions of the form $\neg(\alpha\impl\beta)$ we could choose $\formsize{\neg(\alpha\impl\beta)} = \formsize{\alpha} + \formsize{\beta} + 3$ as a the degree of such a cut. But with that choice, we would run into a problem elsewhere! When we graft a proof $\rho$ with conclusion $\alpha$ on to a proof $\pi$, we might end up creating cuts involving $\rnraa{\alpha}$, and we need to ensure that the cut rank after grafting remains within bounds. Careful consideration leads us to choose $\formsize{\alpha} + \formsize{\beta} + 2$ as a happy medium. Similarly, we choose the degree of the cut $\rnraa{\bot}(\mu)$ to be $2$, and that of $\rnraa{\neg\alpha}$ to be $\formsize{\neg\neg\alpha} = \formsize{\alpha} + 4$.\footnote{It is possible that there is some degree of freedom and that these numbers can be tweaked without harming any of the proofs concerning the decrease of cut ranks. But these choices work, and we stop there. Another interesting point -- these considerations came to the fore only while trying to formalize the proof in Lean. In earlier versions of our proof, we had the wrong definition of $\degof$ and our proofs were wrong. Thus, the Lean formalization has been of immense help in even discovering these measures and getting these proofs right, much less reaffirming our faith in them.} 

We hope that the above discussion adequately motivates the (quite involved) transformations which are formally presented in the next section. 

\section{Normalization for \texorpdfstring{$\classlogic$}{\textsf{NK}}}~\label{sec:nkNorm}
As mentioned earlier, there are no cuts whose main proof end in $\rncase$, and there are new cuts involving $\rnraa{}$ instead of $\rnemp{}$. We first present the definition of cuts, degree, and cutrank. 
\begin{definition}[Cuts and cut-rank]
    The following table lists the patterns that define a cut $\mu$, along with its \defemph{cut-formula} (denoted $\cutform(\mu)$), \defemph{degree} (denoted $\degree(\mu)$), and \defemph{weight} (denoted $\cutwt(\mu)$).
    \begin{longtable}{C|C|C|C|C}
        \mu & \text{Preonditions} & \text{Cut-formula} & \text{Degree} & \text{Weight} \\
        \hline 
        \hline
        \rnraa{\bot}(\nu) & \text{None} & \neg\bot & 2 & 1 \\ 
        \rnraa{\neg\alpha}(\nu) & \text{None} & \neg\neg\alpha & \formsize{\alpha} + 4 & 1 \\ 
        \hline
        \multirow{2}{*}{$\rnrule(\nu, \ldots)$} & \rnrule \in \{\rnapp, \rnfst, \rnsnd, \rncase\} & \multirow{2}{*}{$\concof(\nu)$} & \multirow{2}{*}{$\formsize{\concof(\nu)}$} & \multirow{2}{*}{$1$} \\
        & \lrof(\nu) \in \{\rnabs{}, \rnpair, \rnleft{}, \rnright{}\} & & & \\
        \hline        
        \multirow{2}{*}{$\rnrule(\nu, \ldots)$} & \rnrule \in \{\rnapp, \rnfst, \rnsnd, \rncase\} & \multirow{2}{*}{$\neg\alpha$} & \multirow{2}{*}{$\formsize{\alpha} + 2$} & \multirow{2}{*}{$1$} \\
        & \lrof(\nu)  = \rnraa{\alpha} & & & \\
    \end{longtable}
    If $\mu$ is not a cut, we still define $\degree(\mu)$ and $\cutwt(\mu)$ to be $0$. The \defemph{cut-rank} of a proof $\pi$, denoted $\cutrank(\pi)$, is defined (as before) to be $(0,0)$ when $\pi = \rnax{}$, and $(\degree(\pi),\cutwt(\pi)) \rankadd \cutrank(\pi_{1}) \rankadd \cdots \rankadd \cutrank(\pi_{k})$ when $\pi = \rnrule(\pi_{1}, \ldots, \pi_{k})$. A proof $\pi$ is \defemph{normal} if $\cutrank(\pi) = (0,0)$. Note that $\pi$ is normal iff none of its subproofs is a cut.
\end{definition}
The graft operation, $\graft{\pi}{\alpha}{\rho}$, is defined exactly as before. As before, it is easy to see that if $\alpha\notin{X}$, then $\graft{\pi}{\alpha}{\rho} = \pi$, and also that $\assof(\graft{\pi}{\alpha}{\rho}) \subseteq Y \cup (X-\alpha)$ and $\concof(\graft{\pi}{\alpha}{\rho}) = \concof(\pi)$. The statement regarding cut-rank of a graft is modified slightly, but the proof is very similar to the one for $\intlogic$. 

\begin{lemma}[Cut-rank of a graft]\label{lem:graft-cr-nk}
    Suppose $\pi$ and $\rho$ are proofs with $\concof(\rho) = \alpha$, $\cutrank(\pi), \cutrank(\rho) < (d,0)$, and $\formsize{\alpha} + 2 \le d$. Then $\cutrank(\graft{\pi}{\alpha}{\rho}) < (d,0)$.
\end{lemma}

For the transformations involving $\rnraa{}$, we introduce the notion of an \emph{eliminator}.

\begin{definition}[Eliminators]
    An \defemph{eliminator} is any element $\eliminator$ of the set
    {
        \setlength{\abovedisplayskip}{1pt}
        \setlength{\belowdisplayskip}{1pt}
        \[
            \left\{\elimbot, \elimneg{\alpha}, \elimapp{\alpha\impl\beta}{\mu}, \elimfst{\alpha\conj\beta}, \elimsnd{\alpha\conj\beta}, \elimcase{\alpha\disj\beta}{\mu, \nu}\right\}
        \] 
    }
    where:
    \begin{itemize}[nosep]
        \item $\alpha, \beta$ are formulas, and $\mu, \nu$ are proofs.
        \item In $\elimapp{\alpha\impl\beta}{\mu}$, $\concof(\mu) = \alpha$.
        \item In $\elimcase{\alpha\disj\beta}{\mu,\nu}$, $\concof(\mu) = \concof(\nu) = \bot$.
    \end{itemize}
\end{definition}

The idea is that each eliminator can be applied on a proof to obtain another proof with the same conclusion, and with some open assumptions replaced by simpler assumptions. We formalize the application of eliminators (and contraction) first, and then prove the cutrank reduction lemmas. The application of an eliminator $\eliminator$ to a proof $\pi$ is denoted $\elimtrans{\pi}{\eliminator}$, and is defined by pattern matching. To reduce clutter, we adopt the standard convention that a rule on a line applies only if all rules in the previous lines do not apply. For instance, in the table for $\elimbot$, line 2 applies only if $\gamma \neq \neg\bot$, line 4 applies only if $\lrof(\pi_{2}) \neq \rnraa{\bot}$, and line 5 applies only if $\rnrule \neq \rnapp$ or $\lrof(\pi_{1}) \neq \rnax{\neg\bot}$. Similarly for the other tables.

\begin{definition}[Applying $\elimbot$]
    We define $\elimtrans{\pi}{\elimbot}$ in the table below. It can be proved by a routine induction that for any proof $\pi$ of $X \vdash \phi$, $\assof(\elimtrans{\pi}{\elimbot}) \subseteq X - \neg\bot$ and $\concof(\elimtrans{\pi}{\elimbot}) = \phi$. For ease of notation, we use $\overline{\pi}$ to denote $\elimtrans{\pi}{\elimbot}$ in the table below.
    {
        \setcounter{rownum}{0}
        \begin{longtable}{l|C|C}
            & \pi & \overline{\pi} \\ 
            \hline 
            \hline
            \rownum. & \rnax{\neg\bot} & \rnabs{\bot}(\rnax{\bot}) \\             
            \rownum. & \rnax{\gamma} &  \rnax{\gamma} \\ 
            \rownum. & \rnapp(\rnax{\neg\bot}, \rnraa{\bot}(\pi_{2})) & \overline{\pi_{2}} \\ 
            \rownum. & \rnapp(\rnax{\neg\bot}, \pi_{2}) & \overline{\pi_{2}} \\ 
            \rownum. & \rnrule(\pi_{1}, \ldots, \pi_{k}) & \rnrule(\overline{\pi_{1}}, \ldots, \overline{\pi_{k}}) \\
        \end{longtable}
    }
    It can be proved by a routine induction that for any proof $\pi$ of $X \vdash \phi$, $\assof(\overline{\pi}) \subseteq X \setminus \{\neg\bot\}$ and $\concof(\overline{\pi}) = \phi$. 
\end{definition}
\begin{definition}[Applying $\elimneg{\alpha}$]
    We define $\elimtrans{\pi}{\elimneg{\alpha}}$ in the table below. For ease of notation, we use $\overline{\pi}$ to denote $\elimtrans{\pi}{\elimneg{\alpha}}$ in the table below.
    {
        \setcounter{rownum}{0}
        \begin{longtable}{l|C|C}
            & \pi & \overline{\pi} \\ 
            \hline 
            \hline
            \rownum. & \rnax{\neg\neg\alpha} & \rnabs{\neg\alpha}(\rnapp(\rnax{\neg\alpha}, \rnax{\alpha}))\\             
            \rownum. & \rnax{\gamma} & \rnax{\gamma} \\ 
            \rownum. & \rnapp(\rnax{\neg\neg\alpha}, \rnraa{\neg\alpha}(\pi_{2})) & \overline{\pi_{2}} \\ 
            \rownum. & \rnapp(\rnax{\neg\neg\alpha}, \pi_{2}) & \rnapp(\overline{\pi_{2}}, \rnax{\alpha}) \\ 
            \rownum. & \rnrule(\pi_{1}, \ldots, \pi_{k}) & \rnrule(\overline{\pi_{1}}, \ldots, \overline{\pi_{k}}) \\
        \end{longtable}
    }
    It can be proved by a routine induction that for any proof $\pi$ of $X \vdash \phi$, $\assof(\overline{\pi}) \subseteq (X \cup \{\alpha\}) \setminus \{\neg\neg\alpha\}$ and $\concof(\overline{\pi}) = \phi$. 
\end{definition}
\begin{definition}[Applying $\elimapp{\alpha\impl\beta}{\mu}$]
    We define $\elimtrans{\pi}{\elimapp{\alpha\impl\beta}{\mu}}$ in the table below. For ease of notation, we use $\overline{\pi}$ to denote $\elimtrans{\pi}{\elimapp{\alpha\impl\beta}{\mu}}$ in the table below.
    {
        \setcounter{rownum}{0}
        \begin{longtable}{l|C|C}
            & \pi & \overline{\pi} \\ 
            \hline 
            \hline
            \rownum. & \rnax{\neg(\alpha\impl\beta)} & \rnabs{\alpha\impl\beta}(\rnapp(\rnax{\neg\beta}, \rnapp(\rnax{\alpha\impl\beta}, \mu))) \\             
            \rownum. & \rnax{\gamma} & \rnax{\gamma} \\ 
            \rownum. & \rnapp(\rnax{\neg(\alpha\impl\beta)}, \rnraa{\alpha\impl\beta}(\pi_{2})) & \overline{\pi_{2}} \\ 
            \rownum. & \rnapp(\rnax{\neg(\alpha\impl\beta)}, \pi_{2}) & \rnapp(\rnax{\neg\beta}, \rnapp(\overline{\pi_{2}}, \mu)) \\ 
            \rownum. & \rnrule(\pi_{1}, \ldots, \pi_{k}) & \rnrule(\overline{\pi_{1}}, \ldots, \overline{\pi_{k}}) \\
        \end{longtable}
    }
    It can be proved by a routine induction that for any proof $\pi$ of $X \vdash \phi$, $\assof(\overline{\pi}) \subseteq (X \setminus \{\neg(\alpha\impl\beta)\}) \cup \assof(\mu) \cup \{\neg\beta\}$ and $\concof(\overline{\pi}) = \phi$. 
\end{definition}
\begin{definition}[Applying $\elimfst{\alpha\conj\beta}$]
    We define $\elimtrans{\pi}{\elimfst{\alpha\conj\beta}}$ in the table below. For ease of notation, we use $\overline{\pi}$ to denote $\elimtrans{\pi}{\elimfst{\alpha\conj\beta}}$ in the table below.
    {
        \setcounter{rownum}{0}
        \begin{longtable}{l|C|C}
            & \pi & \overline{\pi} \\ 
            \hline 
            \hline
            \rownum. & \rnax{\neg(\alpha\conj\beta)} & \rnabs{\alpha\conj\beta}(\rnapp(\rnax{\neg\alpha}, \rnfst(\rnax{\alpha\conj\beta}))) \\             
            \rownum. & \rnax{\gamma} & \rnax{\gamma} \\ 
            \rownum. & \rnapp(\rnax{\neg(\alpha\conj\beta)}, \rnraa{\alpha\conj\beta}(\pi_{2})) & \overline{\pi_{2}} \\ 
            \rownum. & \rnapp(\rnax{\neg(\alpha\conj\beta)}, \pi_{2}) & \rnapp(\rnax{\neg\alpha}, \rnfst(\overline{\pi_{2}})) \\ 
            \rownum. & \rnrule(\pi_{1}, \ldots, \pi_{k}) & \rnrule(\overline{\pi_{1}}, \ldots, \overline{\pi_{k}}) \\
        \end{longtable}
    }
    It can be proved by a routine induction that for any proof $\pi$ of $X \vdash \phi$, $\assof(\overline{\pi}) \subseteq (X \cup \{\neg\alpha\}) \setminus \{\neg(\alpha\conj\beta)\}$ and $\concof(\overline{\pi}) = \phi$. 
\end{definition}
\begin{definition}[Applying $\elimsnd{\alpha\conj\beta}$]
    We define $\elimtrans{\pi}{\elimsnd{\alpha\conj\beta}}$ in the table below. For ease of notation, we use $\overline{\pi}$ to denote $\elimtrans{\pi}{\elimsnd{\alpha\conj\beta}}$ in the table below.
    {
        \setcounter{rownum}{0}
        \begin{longtable}{l|C|C}
            & \pi & \overline{\pi} \\ 
            \hline 
            \hline
            \rownum. & \rnax{\neg(\alpha\conj\beta)} & \rnabs{\alpha\conj\beta}(\rnapp(\rnax{\neg\beta}, \rnsnd(\rnax{\alpha\conj\beta}))) \\             
            \rownum. & \rnax{\gamma} & \rnax{\gamma} \\ 
            \rownum. & \rnapp(\rnax{\neg(\alpha\conj\beta)}, \rnraa{\alpha\conj\beta}(\pi_{2})) & \overline{\pi_{2}} \\ 
            \rownum. & \rnapp(\rnax{\neg(\alpha\conj\beta)}, \pi_{2}) & \rnapp(\rnax{\neg\beta}, \rnsnd(\overline{\pi_{2}})) \\ 
            \rownum. & \rnrule(\pi_{1}, \ldots, \pi_{k}) & \rnrule(\overline{\pi_{1}}, \ldots, \overline{\pi_{k}}) \\
        \end{longtable}
    }
    It can be proved by a routine induction that for any proof $\pi$ of $X \vdash \phi$, $\assof(\overline{\pi}) \subseteq (X \cup \{\neg\beta\}) \setminus \{\neg(\alpha\conj\beta)\}$ and $\concof(\overline{\pi}) = \phi$. 
\end{definition}
\begin{definition}[Applying $\elimcase{\alpha\disj\beta}{\mu,\nu}$]
    We define $\elimtrans{\pi}{\elimcase{\alpha\disj\beta}{\mu,\nu}}$ in the table below. For ease of notation, we use $\overline{\pi}$ to denote $\elimtrans{\pi}{\elimcase{\alpha\disj\beta}{\mu,\nu}}$ in the table below.
    {
        \setcounter{rownum}{0}
        \begin{longtable}{l|C|C}
            & \pi & \overline{\pi} \\ 
            \hline 
            \hline
            \rownum. & \rnax{\neg(\alpha\disj\beta)} & \rnabs{\alpha\disj\beta}(\rncase(\rnax{\alpha\disj\beta}, \mu, \nu)) \\             
            \rownum. & \rnax{\gamma} & \rnax{\gamma} \\ 
            \rownum. & \rnapp(\rnax{\neg(\alpha\disj\beta)}, \rnraa{\alpha\disj\beta}(\pi_{2})) & \overline{\pi_{2}} \\ 
            \rownum. & \rnapp(\rnax{\neg(\alpha\disj\beta)}, \pi_{2}) & \rncase(\overline{\pi_{2}}, \mu, \nu) \\ 
            \rownum. & \rnrule(\pi_{1}, \ldots, \pi_{k}) & \rnrule(\overline{\pi_{1}}, \ldots, \overline{\pi_{k}}) \\
        \end{longtable}
    }
    It can be proved by a routine induction that for any proof $\pi$ of $X \vdash \phi$, $\assof(\overline{\pi}) \subseteq (X \setminus \{\neg(\alpha\disj\beta)\}) \cup (\assof(\mu)\setminus\{\alpha\}) \cup (\assof(\nu)\setminus\{\beta\})$ and $\concof(\overline{\pi}) = \phi$. 
\end{definition}
\begin{definition}[Contractum]\label{def:ctm-nk}
    For any cut $\pi$, we define its \defemph{contractum}, denoted $\contractum(\pi)$. 
    {\small
    \setcounter{rownum}{0}
    \begin{longtable}{l|CC|C}
        & \pi & & \contractum(\pi) \\ 
        \hline
        \rownum. & \rnapp(\rnabs{\alpha}(\mu), \nu) & & \graft{\mu}{\alpha}{\nu} \\ 
        \rownum. & \rnfst(\rnpair(\mu, \nu)) & & \mu \\ 
        \rownum. & \rnsnd(\rnpair(\mu, \nu)) & & \nu \\ 
        \rownum. & \rncase(\rnleft{\beta}(\chi), \mu, \nu) & (\concof(\chi) = \alpha) & \graft{\mu}{\alpha}{\chi} \\ 
        \rownum. & \rncase(\rnright{\alpha}(\chi), \mu, \nu) & (\concof(\chi) = \beta) & \graft{\nu}{\beta}{\chi} \\ 
        \rownum. & \rnraa{\bot}(\mu) & & \elimtrans{\mu}{\elimbot} \\ 
        \rownum. & \rnraa{\neg\neg\alpha}(\mu) & & \rnabs{\alpha}(\elimtrans{\mu}{\elimneg{\alpha}}) \\ 
        \rownum. & \rnapp(\rnraa{\alpha\impl\beta}(\mu), \nu) & & \rnraa{\beta}\left(\elimtrans{\mu}{\elimapp{\alpha\impl\beta}{\nu}}\right) \\ 
        \rownum. & \rnfst(\rnraa{\alpha\conj\beta}(\mu)) & & \rnraa{\alpha}\left(\elimtrans{\mu}{\elimfst{\alpha\conj\beta}}\right) \\ 
        \rownum. & \rnsnd(\rnraa{\alpha\conj\beta}(\mu)) & & \rnraa{\beta}\left(\elimtrans{\mu}{\elimsnd{\alpha\conj\beta}}\right) \\ 
        \rownum. & \rncase(\rnraa{\alpha\disj\beta}(\chi), \mu, \nu) & & \elimtrans{\chi}{\elimcase{\alpha\disj\beta}{\mu,\nu}} \\
    \end{longtable}
    }
\end{definition}

We now show that applying eliminators and contracting both lead to a decrease in cut rank, under appropriate conditions. Where the structure of the induction is very similar to the intuitionistic logic case, we skip the proofs. We first state and prove the statement related to applying $\elimapp{\alpha\impl\beta}{\mu}$. There are similar statements we can prove for other eliminators, but we present only this as a representative case. 
\begin{lemma}[Cut-rank on applying $\elimapp{\alpha\impl\beta}{\mu}$]\label{lem:cr-red-elimapp}
    Suppose $d \in \natural$, $\alpha, \beta$ are formulas and $\mu$ a proof of $\alpha$ s.t.\ $\formsize{\alpha\impl\beta} < d$ and $\cutrank(\mu) < (d,0)$. For any proof $\pi$, if $\cutrank(\pi) < (d,0)$, then $\cutrank(\elimtrans{\pi}{\elimapp{\alpha\impl\beta}{\mu}}) < (d,0)$.
\end{lemma}
\begin{proof}
    Throughout this proof, we will refer to $\elimtrans{\pi}{\elimapp{\alpha\impl\beta}{\mu}}$ as $\overline{\pi}$, for ease of notation. We first show by a routine induction that for all proofs $\pi$, we have:
    \begin{equation}
        \lrof(\overline{\pi}) = \lrof(\pi) \text{ or } (\pi = \rnax{\neg(\alpha\impl\beta)} \text{ and } \lrof(\overline{\pi}) = \rnabs{\alpha\impl\beta}) \text{ or } \concof(\pi) = \concof(\overline{\pi}) = \bot \tag{$\dagger$}. \label{eq:dagger}
    \end{equation}

    The main proof is by induction on the structure of $\pi$. 
    \begin{itemize}[nosep]
        \item Suppose $\pi = \rnax{\neg(\alpha\impl\beta)}$. Then $\overline{\pi}$ is $\rnabs{\alpha\impl\beta}(\rnapp(\rnax{\neg\beta}, \rnapp(\rnax{\alpha\impl\beta}, \mu)))$. Any cut in $\overline{\pi}$ occurs inside $\mu$, and its degree is $< d$ by assumption. So $\cutrank(\overline{\pi}) < (d,0)$.
        \item Suppose $\pi = \rnax{\gamma}$ with $\gamma \neq \neg(\alpha\impl\beta)$. In this case $\overline{\pi}$ is the normal proof $\rnax{\gamma}$. 
        \item Suppose $\pi = \rnapp(\rnax{\neg(\alpha\impl\beta)}, \rnraa{\alpha\impl\beta}(\pi_{2}))$, in which case $\overline{\pi} = \overline{\pi_{2}}$. Now $\pi_{2}$ has cut-rank $< (d,0)$ by assumption and $\sizeof(\pi_{2}) < \sizeof(\pi)$. Therefore $\overline{\pi_{2}}$ (and hence $\overline{\pi}$) has cut-rank $< (d,0)$ by IH.  
        \item Suppose $\pi = \rnapp(\rnax{\neg(\alpha\impl\beta)}, \pi_{2})$, with $\lrof(\pi_{2}) \neq \rnraa{}$. 
        In this case, $\overline{\pi}$ is $\rnapp(\rnax{\neg\beta}, \rnapp(\overline{\pi_{2}}, \mu))$. The cuts occurring in $\overline{\pi}$ are those occurring inside $\overline{\pi_{2}}$ (which have degree $< d$ by induction hypothesis) or those occurring inside $\mu$ (which have degree $< d$ by assumption) or $\rnapp(\overline{\pi_{2}}, \mu)$ itself. In the last case, the degree of the cut is $\formsize{\alpha\impl\beta} < d$, unless $\lrof(\overline{\pi_{2}}) = \rnraa{\alpha\impl\beta}$. But in such a case, $\lrof(\pi_{2}) = \lrof(\overline{\pi_{2}})$, which contradicts our assumption above that $\lrof(\pi_{2}) \neq \rnraa{}$.
        \item Suppose $\pi = \rnrule(\pi_{1}, \ldots, \pi_{k})$ where $\rnrule \neq \rnapp$ or $\pi_{1} \neq \rnax{\neg(\alpha\impl\beta)}$. We have $\overline{\pi} = \rnrule(\overline{\pi_{1}}, \ldots, \overline{\pi_{k}})$ and each $\overline{\pi_{i}}$ has cut-rank $< (d,0)$, by induction hypothesis. We consider the following two cases. 
        \begin{description}[nosep]
            \item[$\pi_{1} \neq \rnax{\neg(\alpha\impl\beta)}$:] From~(\ref{eq:dagger}) it follows that $\lrof(\pi_{1}) = \lrof(\overline{\pi_{1}})$, and we see that $\overline{\pi}$ is a cut iff $\pi$ is, and their degree is $< d$, by assumption on $\pi$. 
             \item[$\pi_{1} = \rnax{\neg(\alpha\impl\beta)}$:] From~(\ref{eq:dagger}), we have that $\lrof(\overline{\pi_{1}}) = \rnabs{\alpha\impl\beta}$. It is not possible to have $\rnrule \in \{\rnraa{}, \rnfst, \rnsnd, \rncase\}$ when the conclusion of the main subproof is an implication. We also have $\rnrule \neq \rnapp$, and thus $\overline{\pi}$ is not a cut.
        \end{description}
        Thus $\overline{\pi}$ itself is either not a cut or is a cut of degree $< d$, and all its proper subproofs have cut-rank $< (d,0)$. Thus we have that $\cutrank(\overline{\pi}) < (d,0)$.
        \qedhere
    \end{itemize}
\end{proof}

We now state the important statement that under appropriate conditions, applying a contraction leads to a decrease in cut-rank. We omit the proof, as the proof for cuts not involving $\rnraa{}$ are the same as for $\intlogic$ (and we in fact avoid all complications arising from have the main subproof of a cut ending in $\rncase$). For cuts involving $\rnraa{}$, we appeal Lemma~\ref{lem:cr-red-elimapp} (and its companions) to prove the desired result. 
\begin{lemma}[Cut-rank of a contractum]\label{lem:cr-red-ctm-nk}
    Suppose $\pi$ is a cut of degree $d$, and all its ancilliary subproofs have cut-rank $< (d,0)$. Suppose further that $\lrof(\upsilon) = \rncase$ or $\cutrank(\upsilon) < (d,0)$, where $\upsilon$ is the main subproof of $\pi$. Then $\cutrank(\contractum(\pi)) < \cutrank(\pi)$. 
\end{lemma}

The next step after defining contraction is to define one-step reduction. The situation is considerably simplified since $\rncase$ cannot be the main subproof of a cut, and so one-step reduction is essentially ``contract an uppermost cut of degree $d$, where the proof has cut-rank $(d,\dot{})$''.
\begin{definition}[One-step reduction]\label{def:one-step-red-nk}
    For proofs $\pi$ and $\overline{\pi}$ with $\cutrank(\pi) = (d,n)$, we say that $\pi \oneStep \overline{\pi}$ iff one of the conditions in the table below holds:
    \setcounter{rownum}{0}
    \begin{longtable}{l|C|C|C}
        & \pi & \text{Preconditions} & \overline{\pi} \\ 
        \hline 
        \hline
        \rownum. & \rnrule(\pi_{1}, \ldots, \pi_{i}, \ldots, \pi_{k}) & \degree(\pi) = d,\quad \cutrank(\pi_{1}), \ldots, \cutrank(\pi_{k}) < (d,0), & \contractum(\pi) \\ 
        \hline  
        \rownum. & \rnrule(\pi_{1}, \ldots, \pi_{k}) & \cutrank(\pi_{i}) \ge (d,0),\quad \pi_{i} \oneStep \overline{\pi_{i}} & \rnrule(\pi_{1}, \ldots, \overline{\pi_{i}}, \ldots, \pi_{k}) \\  
        \end{longtable}
\end{definition}
The equivalents of Lemma~\ref{lem:red-exists}, Lemma~\ref{lem:cr-red-red}, and finally Theorem~\ref{thm:normal} are proved very similarly. We restate the main theorem, for the sake of completeness.

\begin{theorem}[Normalization]\label{thm:normal-nk}
    If\, $X \classDerives \alpha$, there is a normal proof of\, $Y \vdash \alpha$ for some $Y \subseteq X$. 
\end{theorem}

The final order of business is to prove a version of the Subformula Property. Since the $\rnraa{\alpha}$ rule discharges instances of the assumption $\neg\alpha$, we cannot work with the definition of $\sfof$ as given earlier. We need to work with $\clof(\phi)$. We first define $\overline{\alpha}$ as follows:
\[
    \overline{\alpha} \coloneq 
    \begin{cases}
        \beta & \text{if }\alpha = \neg\beta \\ 
        \neg\alpha & \text{if }\alpha\text{ is not of the form }\neg\beta \\
    \end{cases}
\]
\begin{itemize}[nosep]
    \item $\clof(p) \coloneq \{p,\overline{p},\bot,\overline{\bot}\}$. 
    \item $\clof(\bot) \coloneq \{\bot,\overline{\bot}\}$. 
    \item $\clof(\alpha\circ\beta) \coloneq \{\alpha\circ\beta, \overline{\alpha\circ\beta}\} \cup \clof(\alpha) \cup \clof(\beta)$.
\end{itemize}

\begin{theorem}[Subformula Property]
    If $\pi$ is a normal proof, then $\formsof(\pi) \subseteq \clof(\assof(\pi) \cup \{\concof(\pi)\})$. Furthermore, if $\pi$ is neutral, then $\formsof(\pi) \subseteq \clof(\assof(\pi))$. 
\end{theorem} 

The subformula property is proved by induction, similar to the corresponding theorem for $\intlogic$. We use the fact that in a normal proof, we do not discharge assumptions of the form $\neg\neg\alpha$ (unless $\neg\neg\alpha$ is part of the conclusion). We have also defined $\clof$ carefully so as not add double negations unless they are already part of the open assumptions or conclusion. 

\section{Conclusion}~\label{sec:conc}
We can adapt our proofs to first-order logic. For intuitionistic first-order logic, if $\pi$ is a cut with main subproof $\mu$, and the last rule of $\mu$ is $\exists$-elimination, then we define $\degree(\pi) = \formsize{\concof(\mu)}$ and $\cutwt(\pi) = \sizeof(\mu)$. For the one step rule, we treat $\exists$-elimination like $\rncase$. The universal quantifier is treated as a generalization of conjunction. For classical first-order logic, we consider a restricted $\exists$-elimination where the conclusion is always $\bot$. We leave the details to be worked out in future, but believe that our proof approach would extend to first-order logic without much change. 

In conclusion, we have presented what we believe to be new and succinct proofs of fundamental results in proof theory -- weak normalization for intutionistic and classical propositional logic. The ideas in our proofs are implicit in earlier versions, but we have not encountered an explicit formulation matching ours. We believe that our proofs can be covered in a lecture or two in a proof theory course, almost in its entirety. We hope that this work contributes to the pedagogy of proof theory, and reaches a wider audience. The Lean formalization lends more credence to our proof, and is of interest in its own right. It has the potential to be an ingredient in the formalization of various theories built on the top of $\intlogic$ or $\classlogic$.

\end{document}